\documentclass[sigconf,nonacm]{aamas}
\setcopyright{none}

\usepackage{balance} 
\usepackage[capitalize,noabbrev]{cleveref}

\usepackage{amsmath,amsfonts,bm}

\def\eqref#1{equation~\ref{#1}}

\def\1{\bm{1}}

\DeclareMathAlphabet{\mathsfit}{\encodingdefault}{\sfdefault}{m}{sl}
\SetMathAlphabet{\mathsfit}{bold}{\encodingdefault}{\sfdefault}{bx}{n}

\DeclareMathOperator*{\argmax}{arg\,max}
\DeclareMathOperator*{\argmin}{arg\,min}

\newcommand{\dd}{{\rm d}}

\usepackage{diagbox}
\usepackage{caption}
\usepackage{wrapfig}
\usepackage{microtype}
\usepackage{graphicx}
\usepackage{subfigure}
\usepackage{subcaption}
\usepackage{booktabs} 
\usepackage{titlesec} 
\usepackage{etoolbox}
\usepackage{tikz}

\usepackage{algorithm}
\usepackage[noend]{algorithmic}
\renewcommand{\algorithmiccomment}[1]{\bgroup//~#1\egroup}

\makeatletter
\newcommand{\newreptheorem}[2]{
	\newtheorem*{rep@#1}{\rep@title}
	\newenvironment{rep#1}[1]
		{\def\rep@title{#2 \ref*{##1}}\begin{rep@#1}}
		{\end{rep@#1}}
}
\makeatother

\newtheorem{theorem}{Theorem}[section]
\newreptheorem{theorem}{Theorem}        

\newtheorem{lemma}[theorem]{Lemma}

\newtheorem{fact}[theorem]{Fact}
\theoremstyle{definition}
\newtheorem{definition}[theorem]{Definition}
\newtheorem{assumption}[theorem]{Assumption}
\theoremstyle{remark}

\title{Towards Global Optimality in Cooperative MARL with the Transformation And Distillation Framework}

\author{Jianing Ye}
\affiliation{
  \institution{Washington University in St. Louis}
  \city{St. Louis}
  \country{the United States}}
\email{jianing.ye@wustl.edu}

\author{Chenghao Li}
\affiliation{
  \institution{Institute for Interdisciplinary Information Sciences, Tsinghua University}
  \city{Beijing}
  \country{China}}
\email{lich18@mails.tsinghua.edu.cn}

\author{Yongqiang Dou}
\affiliation{
  \institution{Department of Computer Science and Technology, Tsinghua University}
  \city{Beijing}
  \country{China}}
\email{dyq21@mails.tsinghua.edu.cn}
\author{Jianhao Wang}
\affiliation{
  \institution{Institute for Interdisciplinary Information Sciences, Tsinghua University}
  \city{Beijing}
  \country{China}}
\email{wjh19@mails.tsinghua.edu.cn}

\author{Guangwen Yang}
\affiliation{
  \institution{Department of Computer Science and Technology, Tsinghua University}
  \city{Beijing}
  \country{China}
}
\email{ygw@tsinghua.edu.cn}

\author{Chongjie Zhang}
\affiliation{
  \institution{Washington University in St. Louis}
  \city{St. Louis}
  \country{the United States}}
\email{chongjie@wustl.edu}

\begin{abstract}
  Decentralized execution is  one core demand in multi-agent reinforcement learning (MARL). 
  Recently, most popular MARL algorithms have adopted decentralized policies to enable decentralized execution, and use gradient descent as the optimizer.
  However, there is hardly any theoretical analysis of these algorithms taking the optimization method into consideration, and we find that various popular MARL algorithms with decentralized policies are suboptimal in toy tasks when gradient descent is chosen as their optimization method.
  In this paper, we theoretically analyze two common classes of algorithms with decentralized policies -- multi-agent policy gradient methods and value-decomposition methods, and prove their suboptimality when gradient descent is used.
  To address the suboptimality issue, we propose the Transformation And Distillation (TAD) framework, which reformulates a multi-agent MDP as a special single-agent MDP with a sequential structure and enables decentralized execution by distilling the learned policy on the derived "single-agent" MDP.
  The approach is a two-stage learning paradigm that addresses the optimization problem in cooperative MARL, providing optimality guarantee with decent execution performance.
  Empirically, we implement TAD-PPO based on PPO, which can theoretically perform optimal policy learning in the finite multi-agent MDPs and shows significant outperformance on a large set of cooperative multi-agent tasks, from matrix game, hallway task, to StarCraft II, and football game. 
\end{abstract}


\newcommand{\BibTeX}{\rm B\kern-.05em{\sc i\kern-.025em b}\kern-.08em\TeX}

\begin{document}


\pagestyle{fancy}
\fancyhead{}


\maketitle
\let\thefootnote\relax\footnotetext{Preprint.} 


\section{Introduction}
\label{sec:intro}
Cooperative multi-agent reinforcement learning (MARL) is a promising approach to a variety of real-world applications, such as sensor networks \citep{zhang2011coordinated}, traffic light control \citep{van2016coordinated}, and multi-robot formation \citep{alonso2017multi}.
One major challenge in MARL is the scalability issue when the number of agents grows.
To address this challenge, the paradigm of centralized training with decentralized execution (CTDE) \citep{Frans2011Decen,Kraemer2016MultiagentRL} addresses this issue by permitting agents the access of global information during training, while executing in a decentralized manner without communication.
On the other hand, gradient descent and its variants are broadly used in deep MARL, usually as the default optimizer to optimize the loss function.

Recently, many CTDE algorithms have been proposed to enable decentralized execution by directly learning decentralized policies with gradient descent as the optimizer.
COMA \citep{Foerster2018CounterfactualMP} and MAPPO \citep{yu2021surprising} factorize the joint policy as the direct product of decentralized policies. Then they use their critic networks to estimate the policy gradient and update the policy parameters.
VDN \citep{Sunehag2018ValueDecompositionNF}, QMIX \citep{Rashid2018QMIXMV}, and QPLEX \citep{Wang2021QPLEXDD} factorize the joint Q-function into individual Q-functions with the \textit{Individual-Global-Max} (IGM) constraint \citep{Son2019QTRANLT} to enable standard TD-learning, where the decentralized policies are the greedy policies w.r.t. the individual Q-functions.
However, there is barely any theoretical analysis of CTDE algorithms taking the optimization method into consideration, and we find that many of them would get stuck in some local minima and thus lose their optimality guarantee in toy tabular tasks.

To characterize this problem, we theoretically analyze two sub-classes of state-of-the-art cooperative MARL: (1) multi-agent policy gradient algorithms (MA-PG)  \citet{Foerster2018CounterfactualMP,yu2021surprising}; and (2) value-decomposition algorithms (VD) \citet{Sunehag2018ValueDecompositionNF, Rashid2018QMIXMV, Wang2021QPLEXDD}.
We prove that these two classes of algorithms are suboptimal with mild assumptions.
Intuitively, this is because directly learning decentralized policies by MA-PG or VD restricts the parameter space to have a ``decentralized structure", which would create numerous local optima, while gradient descent, as a local search in the parameter space, is hard to escape from a local optimum.

To further address this optimality problem, we present a Transformation And Distillation (TAD) framework. The TAD framework is two-stage. The first stage of TAD reformulates a multi-agent MDP as an equivalent single-agent MDP with a sequential decision-making structure.
With this transformation, any off-the-shelf single-agent reinforcement learning (SARL) algorithm can be readily adapted to efficiently learn a coordination policy in the cooperative multi-agent task by solving the transformed single-agent task.
In the second stage, decentralized policies are derived through distillation from the learned coordination policy of the transformed single-agent task, enabling decentralized execution.
We formally prove that the TAD framework maintains the performance guarantee of its chosen single-agent RL algorithm so that TAD can guarantee optimality with the ability of decentralized execution.

Empirically, we implement TAD-PPO based on PPO \citep{Schulman2017ProximalPO}, the single agent algorithm.
Building upon the TAD framework, TAD-PPO can theoretically perform optimal policy learning in finite multi-agent MDPs (\cref{thm:TAD-PPO-optimality}) and empirically shows significant improvements on a large set of cooperative multi-agent tasks, from one-step matrix game, multi-step hallway task, to SMAC \citep{Samvelyan2019TheSM}, SMAC v2 \citep{ellis2022smacv2}, and GRF \citep{kurach2019google}. We plan to open-source the full algorithm code, and evaluation tasks, to ensure reproducibility and facilitate future works from broader community. We attach supplementary material of proofs and experimental details as Appendix to this paper.

\section{Related Work}
\label{sec:relatedwork}

\paragraph{Multi-Agent
Policy Gradient algorithms (MA-PG).} Many CTDE-based algorithms are proposed to enable decentralized execution in cooperative MARL.
Among them, policy-based algorithms try to optimize the policy by maximizing its value.
Typically, the joint policy is factorized as a direct product of decentralized policies, and is end-to-end updated by the policy gradient \citep{Foerster2018CounterfactualMP,Wang2021DOPOM,lowe2017multi,yu2021surprising}.
Further, there are also algorithms using a heterogeneous way to update the decentralized policies \citep{kuba2021trust,hetero2kuba}, which alleviates the non-stationarity in policy updates.

\paragraph{Value-Decomposition algorithms
(VD) and IGM.} Value-based algorithms try to optimize the Bellman error of the joint value function, which is usually decomposed into a function of local Q-functions. To enable TD-learning, it is crucial to follow the IGM principle \citep{Son2019QTRANLT}.
Various works \citep{Sunehag2018ValueDecompositionNF,Rashid2018QMIXMV} are made to satisfy the IGM principle by giving up part of the completeness of the Q-function class.
Since an incomplete Q-function class would constitute inherent Bellman error \citep{Munos2008FiniteTimeBF} and harm the performance, \citet{Son2019QTRANLT,WQMIX} use soft regularizers to approximate the IGM constraint, and \citet{Wang2021QPLEXDD} is the first work using a duplex dueling structure to make the IGM constraint and the completeness constraint both satisfied.

\paragraph{Theoretical understandings on single-agent and multi-agent RL algorithms.} There are numerous theoretical works that analyze the optimality of popular single-agent algorithms taking the optimization method into consideration \citep{Konda1999ActorCriticA,Off-PAC,singletimescaleActorCritic}.
However, in MARL, most of the analysis considers the optimizer as a black box and assumes it can perfectly solve the corresponding optimization problem \citep{Wang2021QPLEXDD,wang2021towards}.
TAD is the first work that theoretically analyzes the suboptimality of two large classes of MARL algorithms when using gradient descent as the optimizer.
And TAD-PPO is the first MARL algorithm with decentralized execution that not only achieves remarkable performance but also has a theoretical guarantee on its optimality when gradient descent is used.

\section{Preliminaries}
\label{sec:preliminary}
\subsection{RL Models}

In MARL, we model a fully cooperative multi-agent task as a Dec-POMDP \citep{oliehoek2016concise} defined by a tuple $\langle\mathcal{S},\mathcal{A},P, \mathcal{\varOmega}, O, r, \gamma,n \rangle$, where $\mathcal{N}=\{1,\cdots,n\}$ is a set of agents and $\mathcal{S}$ is the global state space, $\mathcal{A}$ is the individual action space, and $\gamma$ is a discount factor. At each time step, agent $i \in \mathcal{N}$ has access to the observation $o_i\in\mathcal{\varOmega}$, drawn from the observation function $O(s,i)$. Each agent has an action-observation history $\tau_i\in\mathcal{\varOmega}\times\left(\mathcal{A}\times\mathcal{\varOmega}\right)^*$ and constructs its individual policy $\pi(a_i|\tau_i)$. With each agent $i$ selecting an action $a_i\in\mathcal{A}$, the joint action $\bm{a} \equiv [a_i]_{i=1}^n$ leads to a shared reward $r(s,\bm{a})$ and the next state $s'$ according to the transition distribution $P(s'|s,\bm{a})$. The formal objective of MARL agents is to find a joint policy $\bm{\pi}=\langle \pi_1, \dots, \pi_n\rangle$ conditioned on the joint trajectories $\bm{\tau}\equiv [\tau_i]_{i=1}^n$ that maximizes a joint value function $V^{\bm{\pi}}(s)=\mathbb{E}\left[\sum_{t=0}^{\infty} \gamma^t r_t|s_0=s,\bm{\pi}\right]$. Another quantity in policy search is the joint action-value function $Q^{\bm{\pi}}(s,\bm{a})=r(s,\bm{a})+\gamma \mathbb{E}_{s'}[V^{\bm{\pi}}(s')]$. 

Besides, we also introduce the Multi-agent MDPs (MMDP) \citep{boutilier1996planning} and Markov Decision Process (MDP) \citep{sutton2018reinforcement} as special cases of Dec-POMDP, where an MMDP is a Dec-POMDP with full-observation, and an MDP is an MMDP with only one agent. The formal definition is deferred to \cref{app:RLmodels}.

\subsection{Multi-agent Policy Gradient and Value-Decomposition}

We introduce two popular classes of cooperative multi-agent learning algorithms for our analysis: multi-agent policy gradient (MA-PG) and value-decomposition (VD).
To simplify our theoretical analysis, our discussion based on the assumption that the environment is an MMDP.

\paragraph{Multi-agent Policy Gradient} MA-PG factorizes the joint policy as a direct product of the decentralized policies with parameter $\Theta\in\text{dom}(\Theta)$:

\begin{align}\label{eq:pi-factor}
    \forall s\in \mathcal S:\quad \bm{\pi}(\bm{a}|s;\Theta) = \prod_{i=1}^n \pi_i(a_i|s;\Theta)
\end{align}

Denoting $\mathcal J(\bm \pi)$ as the expected return of policy $\bm\pi$, the policy is updated by the policy gradient \citep{sutton2018reinforcement} with a small learning rate $\alpha$: $\Theta^{(t+1)}\gets \Theta^{(t)} + \alpha \nabla_\Theta \mathcal J(\bm\pi_\Theta)$, which is equivalent to minimize the following loss function by gradient descent. 

\begin{align}\label{eq:mapg-loss}
    \mathcal L(\Theta) = -\mathcal J(\bm \pi_\Theta)
\end{align}

\paragraph{Value-decomposition} VD maintains a joint Q-value function $Q(s,\bm a;\Theta)$ and $n$ local Q-value functions $Q_1(s,a_1;\Theta),\cdots,Q_n(s,a_n;\Theta)$ parameterized by some $\Theta\in \text{dom} (\Theta)$.
Standard TD-learning is then applied with a small learning rate $\alpha$ to update the parameter $\Theta$, which is equivalent to minimizing the following loss function by gradient descent \citep{lowrankapp2021Rozada}:

\begin{align}\label{eq:vd-loss}
\mathcal L(\Theta) = \frac 1 2\mathbb E_{s,\bm a}\left[
     Q(s, \bm a;\Theta) - (\mathcal T Q_{\Theta})(s,\bm a)
\right]^2
\end{align}

where $\mathcal T$ is the Bellman operator.

Meanwhile, the IGM principle \citep{Son2019QTRANLT} is enforced to realize effective TD-learning, which asserts the consistency between joint and local greedy action selections in the joint action-value $Q(s, \bm{a})$ and local action-values $\left[Q_i(s,a_i)\right]_{i=1}^n$, respectively:

\begin{align}\label{eq:igm}
 \forall s:\quad \mathop{\arg\max}_{\bm{a}\in\bm{\mathcal{A}}} Q(s, \bm{a}) \supseteq \bigtimes_{i=1}^n\mathop{\arg\max}_{a_i\in\mathcal{A}} Q_i(s, a_i)
\end{align}

VD enables IGM condition by representing the joint Q-value function as a specific function $f$ of local Q-functions:

\begin{align}\label{eq:value-decomp}
 Q(s, \bm a;\Theta)  = f(Q_1(s,\cdot), \cdots,Q_n(s,\cdot),s ,\bm a;\Theta)
\end{align}

For example, \citet{Sunehag2018ValueDecompositionNF} represents $Q$ as a summation of $[Q_i]_{i=1}^n$; and \citet{Rashid2018QMIXMV} can represent $Q$ as any monotonic function of $[Q_i]_{i=1}^n$ by adjusting the parameter $\Theta$.

In this way, the joint Q-function is ``decomposed" into local Q-functions, and the decentralized policies of VD can be obtained by calculating the greedy actions w.r.t. the local Q-functions.

\section{Suboptimality of MA-PG and VD with Gradient Descent}
\label{sec:motivation}
In this section, we will formally analyze the suboptimality of MA-PG and VD when using gradient descent as the optimizer.
As mentioned before, gradient descent can be viewed as a local search in the parameter space, which is hard to escape from local optima.
Yet, the loss functions in MA-PG and VD adopt different ways to evaluate how good the current decentralized policies are (\cref{eq:mapg-loss,eq:vd-loss}).
We will show that these approaches create numerous local optima in the loss function regardless of the parameterization of the functions.
Consequently, the convergence results of MA-PG and VD depend heavily on the initialization, which contrasts with the results of many classical SARL algorithms, e.g., Value Iteration.

\subsection{Suboptimality of Multi-agent Policy Gradient}

We first present the suboptimality theorem for MA-PG as follows.

\begin{table}[htb]
    \caption{Matrix Game with locally optimal policies: in this case, $(0,0)$, $(1,1)$, and $(2,2)$ are three locally optimal policies in the production of the decentralized policies' spaces.}
    \centering
    \begin{tabular}{|l|l|l|}
    \hline
    10   & -30 & -30  \\ \hline
    -30  & 5   & -30   \\ \hline
    -30  & -30   & 1  \\ \hline

    \end{tabular}

    \label{matgame1}
\end{table}
\begin{theorem}
\label{thm:ma-pg-fails}

There are tasks such that, when the parameter $\Theta$ is initialized in certain region $\Omega$ with positive volume, MA-PG (\cref{eq:pi-factor,eq:mapg-loss}) converges to a suboptimal policy with a small enough learning rate $\alpha$.
\end{theorem}

To serve intuitions about why \cref{thm:ma-pg-fails} holds, we introduce a matrix game here as a didactic example and defer the proof to \cref{app:proofs-mapg}.
The matrix game described in \cref{matgame1} is a $1$-step cooperative game with two players.
These two players (one selects a row and one selects a column simultaneously, $\mathcal A = \{0,1,2\}$) coordinate to select an entry of the matrix representing the joint payoff.
We will show that the corresponding $\Theta$s of the suboptimal policies $(1,1)$ and $(2,2)$ are stationary points (with zero gradients) of the loss function $\mathcal L$ in this case.


We prove it by contradiction, if the gradient of $\mathcal L$ in MA-PG (\cref{eq:mapg-loss}) is non-zero at some $\Theta$, then by the definition of gradient, we can move $\Theta$ a little bit along the negative gradient to obtain a strictly greater expected payoff $\mathcal J$.
However, this is impossible when the current policy is the deterministic policy $(1,1)$ since one cannot increase the probability of $(0,0)$ (a better payoff, $10$) without increasing the probability of $(0,1)$ or $(1,0)$ (an even larger penalization, $-30$) when the policy space is a product space of decentralized policies' spaces.
A similar situation also occurs when the current policy is $(2,2)$.

This example shows that the decentralized structure in policy restricts the structure of parameter space and may create local optima independent of the parameterization if we do not design the loss function carefully.
This makes the optimality guarantee lost when gradient descent is used.

\subsection{Suboptimality of Value-Decomposition}
For VD, the situation is much more complex.
IGM (\cref{eq:igm}) is a vital precondition to enable efficient TD-learning.
Early work of VD uses specific structures to satisfy a sufficient condition of IGM, such as linear structure \citep{Sunehag2018ValueDecompositionNF} or monotonic structure \citep{Rashid2018QMIXMV}.
These approaches sacrifice the completeness of the Q-function class, which would incur inherent Bellman error and deviate the learning process to unexpected behavior \citep{AgarwalNanJiang2019ReinforcementLT}.
Later, \citet{Wang2021QPLEXDD} addresses the incompleteness issue under the IGM constraint by introducing a duplex dueling structure.

However, all these algorithms fail to guarantee optimality again because of the structure of the parameter space they have.
Since VD with an incomplete function class already suffers from inherent Bellman error as mentioned above, we consider VD with a complete function class.
To characterize the suboptimality of VD with a complete function class, we present \cref{thm:vd-fails}.

\begin{theorem}
\label{thm:vd-fails}
For any VD (\cref{eq:value-decomp,eq:vd-loss,eq:igm}) with a complete function class, there are tasks such that, when the parameter $\Theta$ is initialized in certain set $\Omega$, it converges to a suboptimal policy.
\end{theorem}
The proof and a finer characterization are presented in \cref{app:proofs-vd}.

Roughly speaking, \cref{thm:vd-fails} is a consequence of the IGM constraint over decentralized and joint policies.
Because of it, the process of GD would be a local search over the complete value function class keeping the IGM constraint satisfied all along.
Therefore such a searching process over the space looks ``crooked" and ``obstructed" everywhere, making it easy to get stuck into some local optimum.

\section{Transformation and Distillation Framework}
\label{sec:transformation}

\begin{figure*}[tb]
\centering
\includegraphics[width=.9\linewidth]{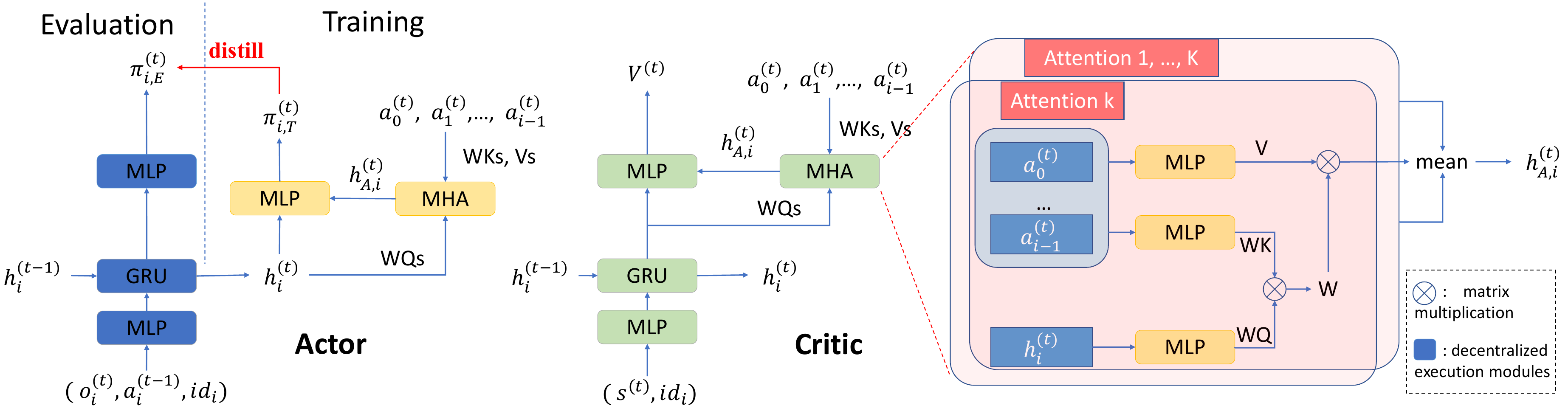}
\caption{
    The architecture of transformation and distillation framework with PPO (TAD-PPO).
    WK, WQ, and V are key matrix, query matrix, and value matrix respectively in the MHA module \citep{attentionisallyouneed}.
}
\label{fig:TAD-PPO_arch}
\end{figure*}

As suggested by \cref{thm:ma-pg-fails,thm:vd-fails}, MA-PG and VD suffer from the cratered loss function induced by the decentralized structure, which makes their convergence results depend heavily on the initialization.
However, the PG with a softmax parameterization and the tabular Q-learning, which are the counterparts of MA-PG and VD in SARL, both guarantee the global convergence to the optimum \citep{softmaxpgglobalconverge,Watkins1992Qlearning} due to their nice global properties of optimization objective.
This observation motivates us to convert the original multi-agent learning problem into a single-agent one to enjoy a possibly better structure in the optimization problem.

Inspired by this observation and a series of prior works on sequence modeling \citep{2019DimitriRollout,Angermueller2020Model-based,2022JainSequenceGFlowNets,2017Molecular,2018MolecularYou}, we find that breaking joint learning into sequential decision-making directly transfers properties of single-agent algorithms into multi-agent ones.
To formalize such ``reduction" of algorithms from the perspective of performance guarantee, we propose the Transformation And Distillation (TAD) framework, which is a two-stage framework under the paradigm of centralized training with decentralized execution.

\subsection{TAD on MMDP}

To simplify our analysis, we assume the underlying task to be an MMDP here.

In the centralized training phase of CTDE, all $n$ agents will coordinate to improve their joint policy.
In particular, when a joint action needs to be inferred, they coordinate to infer $\bm a = (a_1,\cdots,a_n)$ jointly.
If we give agents a virtual order and imagine they infer their individual actions one after another, that is, when agent $i$ infers its action $a_i$, all ``previously inferred" actions $a_j (j<i)$ are known to $i$, then we can view this procedure of multi-agent decision-making to a ``single-agent" one.
From a theoretical perspective, this states a transformation from MMDP into MDP in essence.

We formally define the Sequential Transformation $\Gamma$ as \cref{algdef:transformation}.

\begin{algorithm}[htb]
  \caption{The Sequential Transformation $\Gamma$}
  \label{algdef:transformation}

\begin{algorithmic}[1]
  \STATE {\bfseries Input:} An MMDP $\mathcal M = (\mathcal S,\mathcal A,P,r,\gamma, n)$.
  \STATE {\bfseries Output:} An MDP $\Gamma(\mathcal M) = (\mathcal S', \mathcal A', P',r',\gamma')$.

  \STATE $\mathcal S' \gets \bigcup_{k=0}^{n-1} \mathcal S\times \mathcal A^k$.
  \STATE $\mathcal A' \gets \mathcal A$.
  \STATE $\gamma' = \gamma^{1/n}$. 
  \STATE Initialize $P'(s'|s,a) = 0$ for all $s,s'\in \mathcal S', a\in \mathcal A'$.
  \STATE Initialize $r'(s,a) = 0$ for all $s\in \mathcal S', a\in\mathcal A'$.
  \FORALL{$s'\in \mathcal S'$}
  \IF{$s' = (s,a_1,\cdots,a_k)$ \AND $k<n-1$}
    \STATE\COMMENT{If this is not the last action:}
    \FORALL{$a_{k+1}\in \mathcal A'$}
    \STATE $u \gets (s,a_1,\cdots, a_k, a_{k+1})$.\quad\COMMENT{Create a new virtual state by appending $a_{k+1}$ to $s'$.}
    \STATE $P'(u|s',a_{k+1}) \gets 1$.\quad\COMMENT{Set a deterministic transition from $s'$ to $u$.}
    \ENDFOR
  \ELSE
    \STATE\COMMENT{If this is the last action:}
    \FORALL{$a_n\in \mathcal A'$}
        \STATE $r'(s',a_n) \gets r(s,a_1,\cdots,a_n)$.
        \STATE\COMMENT{Set the reward according to $\mathcal M$.}
        \FORALL{$u\in \mathcal S$}
            \STATE $P'(u|s',a_n) \gets P(u|s,a_1,\cdots,a_n)$.
            \STATE\COMMENT{Set the transition according to $\mathcal M$.}
        \ENDFOR
    \ENDFOR
  \ENDIF
  \ENDFOR

  \STATE {\bfseries Return} $\Gamma(\mathcal M) = (\mathcal S', \mathcal A', P',r',\gamma')$.
\end{algorithmic}
\end{algorithm}

We can then prove the equivalence of the MMDP $\mathcal M$ and the MDP $\Gamma (\mathcal M)$ after transformation in the perspective of policy value by the following theorem.
The proof is presented in \cref{app:transformation}.

\begin{theorem}
\label{thm:convert-policy-keeps-value}
Let $C = \gamma^{(1-n)/n}$ be a constant. For any learned policy $\pi$ on the transformed MDP $\Gamma(\mathcal M)$, we can convert $\pi$ into a coordination policy $\pi_{\text{c}}$ on the original MMDP $\mathcal M$, and vice versa, such that $\mathcal J_{\Gamma(\mathcal M)}(\pi) = C \mathcal J_{\mathcal M} (\pi_c)$.
\end{theorem}

\cref{thm:convert-policy-keeps-value} says that there is an easy-to-compute bijection between the policy space of $\mathcal M$ and that of $\Gamma (\mathcal M)$, such that the relative magnitude of policy value retains.

\paragraph{The Transformation Stage}
In the transformation stage of TAD, we use some SARL algorithm $A$, wrap the interface of the multi-agent environment to make $A$ run as if it is accessible to the interactive environment of $\Gamma\left(\mathcal M\right)$.
Then we convert the policy learned by $A$ to a coordination policy on $\mathcal M$ without loss of performance by \cref{thm:convert-policy-keeps-value}.

\paragraph{The Distillation Stage}
In the distillation stage of TAD, a distillation process is applied to the coordination policy to obtain decentralized policies.
This is easy to be done on MMDP, since we have full observations and therefore can completely know what allies did and will do, which means we calculate the coordination policy locally.
To be succinct, the details are presented in \cref{app:transformation}.

We present the whole TAD framework as \cref{alg:TAD-framework}, which enjoys the following theorem.

\begin{algorithm}[htb]
  \caption{The Transformation And Distillation (TAD) framework}
  \label{alg:TAD-framework}
  
\begin{algorithmic}[1]
  \STATE {\bfseries Input:} An MMDP $\mathcal M$ of $n$ agents, a single agent algorithm $A$.
  \STATE {\bfseries Output:} Decentralized policies $\pi_1,\cdots,\pi_n$.

  \STATE {\bfseries The Transformation Stage (for centralized training):}

  \STATE $\mathcal N\gets \Gamma(\mathcal M)$.\quad\quad\COMMENT{Transform $\mathcal M$ into an MDP $\mathcal N$.}
  \STATE Run $A$ on $\mathcal N$ to obtain a policy $\pi$.
  \STATE Convert $\pi$ to a coordination policy $\pi^c$.

  \STATE {\bfseries The Distillation Stage (for decentralized execution):}
  \STATE Distill $\pi^c$ into decentralized policies $\pi_1,\cdots,\pi_n$.

  \STATE {\bfseries Return} $\pi_1,\cdots,\pi_n$.
\end{algorithmic}
\end{algorithm}

\begin{theorem}
\label{thm:TAD-keeps-optimality}
For any SARL algorithm $A$ with an optimality guarantee on MDP, TAD-$A$, the MARL algorithm using the TAD framework and based on $A$, will have an optimality guarantee on any MMDP.
\end{theorem}

\cref{thm:TAD-keeps-optimality} shows that the TAD framework successfully bridges the gap to optimal decentralized policy learning for CTDE algorithms. The proof is presented in \cref{app:transformation}.

\subsection{Extending TAD to Dec-POMDP}
\label{sec:TAD-PPO-structure}

In practice, many MARL benchmarks are partially observable and require each agent to execute independently.
To apply our algorithms to partial-observable environments with a decentralized execution manner, we need to extend the TAD framework to the setting of Dec-POMDP.

In this paper, we select proximal policy optimization (PPO, \citet{Schulman2017ProximalPO}) as the basic SARL algorithm because a suitable implementation of PPO enjoys the optimality guarantee in MDP \citep{ppooptimal, ppo-clip-optimal}, and its counterparts in the multi-agent setting, MAPPO \citep{yu2021surprising} and IPPO \citep{de2020independent}, have already achieved great success in challenging benchmarks. 

The Actor-Critic structure is shown in \cref{fig:TAD-PPO_arch} with GRU modules to represent partial observations. All agents share parameters of actor and critic for enjoying learning efficiency enhanced by parameter sharing \citep{yu2021surprising, li2021celebrating}. To provide representation diversity across agents, we further add each agent's identity as a part of input~\cite{Rashid2018QMIXMV} as shown in \cref{fig:TAD-PPO_arch}. 

Following the TAD framework, we introduce previous agents' actions to each agent's actor and critic modules, so that a sequential decision-making process is enabled in the centralized training phase.
However, this sequential transformed information increases with respect to the number of agents.
To achieve scalability, we equip agents' actor and critic with a multi-head attention (MHA) module\citep{attentionisallyouneed,iqbal2021randomized} respectively, which contains queries WQs, keys WKs, and values Vs. WQs are outputs of an MLP network with $h_{i}^{(t)}$ (output of GRU) as the input. WKs and Vs are outputs of MLP networks with all previous actions $a_{<i}^{({t})}$ as the input.

For the distillation part that distills the  coordination policy into decentralized policies, we optimize the KL divergence between the decentralized policies and the joint coordination policy, which is known as behavior cloning. Denote the joint policy $\pi_{\text{jt}}(\bm a|\bm \tau) = \prod_i \pi_{i,T}(a_i|\tau_i, a_{<i})$, we have
\begin{align}
&\text{KL}\left(\pi_{\text{jt}}(\cdot|\bm\tau)\middle\|\prod_{i=1}^n\pi_{i,E}(\cdot|\tau_i;\phi)\right) \nonumber\\
=&
\underbrace{
    \mathop{\mathbb E}
        _{\boldsymbol a\sim \pi_{\text{jt}}(\cdot|\bm \tau)}
    \left[-\sum_{i=1}^n\log \pi_{i,E}(a_i|\tau_i;\phi) \right]
}_{
    \text{distillation loss (cross entropy) } \mathcal L(\phi)
}
-\underbrace{
    H( \pi_{\text{jt}}(\cdot|\bm \tau))
}_{
    \text{entropy (a constant)}
}
\end{align}

where $\pi_{i,T}$ and $\pi_{i,E}$ are the coordination policy and the decentralized policy of agent $i$.

Here we share GRU and the representation layer, whose inputs do not contain other agents' information. During the evaluation, only blue modules in Figure \ref{fig:TAD-PPO_arch} are activated for decentralized execution. To distinguish policies in different stages, we name the policy in the transformation stage as TAD-PPO-T and the policy during the distillation stage as TAD-PPO.




\section{Experiments}
\label{exp}

We design experiments to answer the following questions: (1) Can the TAD framework achieve the {globally} optimal policy on MMDP? (2) Can TAD-PPO improve learning efficiency for policy-based MARL algorithms? (3) {Does the performance of the decentralized executing policies by the distillation stage degenerate compared with the joint policy for the transformed single-agent MDP?}

We first use {two didactic examples} to demonstrate the global optimality of our coordination policy, {and the suboptimality of the most common MARL algorithms including} multi-agent value-based methods, i.e., VDN \citep{Sunehag2018ValueDecompositionNF}, QMIX \citep{Rashid2018QMIXMV}, QTRAN \citep{Son2019QTRANLT}, and QPLEX \citep{Wang2021QPLEXDD}, and policy-based methods, i.e., MAPPO \citep{yu2021surprising} and HAPPO \citep{kuba2021trust}. To approach complex {environments}, we conduct experiments on the StarCraft II micromanagement (SMAC) benchmark \citep{Samvelyan2019TheSM}, SMAC v2 \citep{ellis2022smacv2} and Google Research Football (GRF) benchmark \citep{kurach2019google}. Evaluation results are averaged over three random seeds and use 95\% bootstrapped confidence interval to present the variance. Readers can refer to \cref{app:experimental_details} for details about experiments.

\subsection{Didactic Examples}

\label{sec:msmatgame}

\begin{figure}[h]
  \centering
    \includegraphics[width=.8\linewidth]{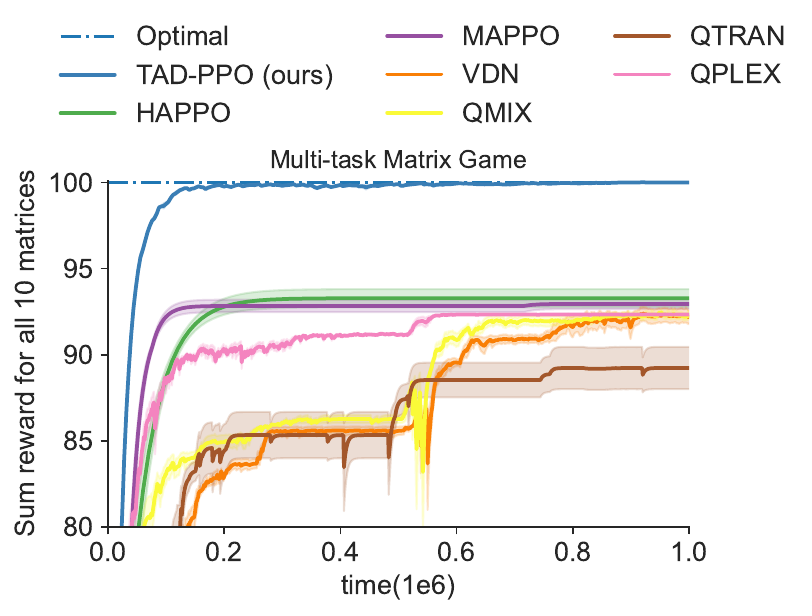}
  \caption{Learning curves in the didactic example.}
  \label{fig:matrix_results}
\end{figure}

We first conduct a didactic example to investigate the optimality {property} of MARL learning algorithms. This didactic example is a stacked matrix game, which has 10 one-step matrices with two players. Details are presented to Appendix \ref{matrices}. Figure~\ref{fig:matrix_results} shows that TAD-PPO achieves the optimal performance and baselines cannot solve these 10 matrices simultaneously. This result suggests that,  our \textit{Transformation And Distillation} framework can {find} the global optimal strategy, with {the optimizer of gradient descent}. As our theory indicates, the MARL policy gradient methods (i.e., {MAPPO and HAPPO}) {converge} to a locally optimal policy due to bad parameter initialization. In MARL value-decomposition methods, VDN and QMIX suffers from incomplete function classes, QPLEX may get stuck in local optima due to its mixing structure, and QTRAN cannot optimize well by its discontinuous loss function. In this didactic example, TAD-PPO achieves state-of-the-art performance with a theoretical guarantee of optimality {under gradient descent}.


\begin{figure*}[t]
  \centering
    \includegraphics[width=.9\linewidth]{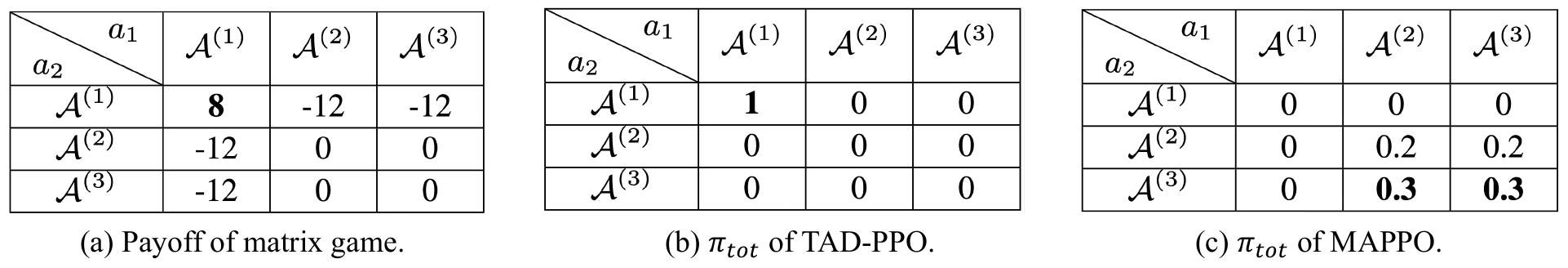}
  \caption{Payoff of one matrix game with local optimal points and learned joint policies based on our approach TAD-PPO and MAPPO.}
  \label{fig:matrix_per}
\end{figure*}

To understand the policy learning in the didactic example, we {show} a commonly-used matrix proposed by QTRAN \citep{Son2019QTRANLT} (\cref{fig:matrix_per} {(a)), and we} visualize the joint policies of TAD-PPO and MAPPO in \cref{fig:matrix_per} (b)(c). As show by the payoff matrix \cref{fig:matrix_per}(a), there are one global optimum point and local optima in this matrix game. Through the \textit{Transformation And Distillation} framework, TAD-PPO can learn the optimal policy (\cref{fig:matrix_per}(b)). However, with random initialization of the neural networks of the policies, MAPPO, one of the most common MARL policy gradient methods, converges to the bad local optima (\cref{fig:matrix_per} (c)).

 \paragraph{Another didactic example \textit{Hallway}, as proofs of concept of the optimality in multi-step scenarios with two-agent cooperation, can be found in \cref{app:hallway_didactic}.}
\subsection{StarCraft II}
\label{smac}

\begin{figure*}[!htbp]
  \centering
    \includegraphics[width=.95\linewidth]{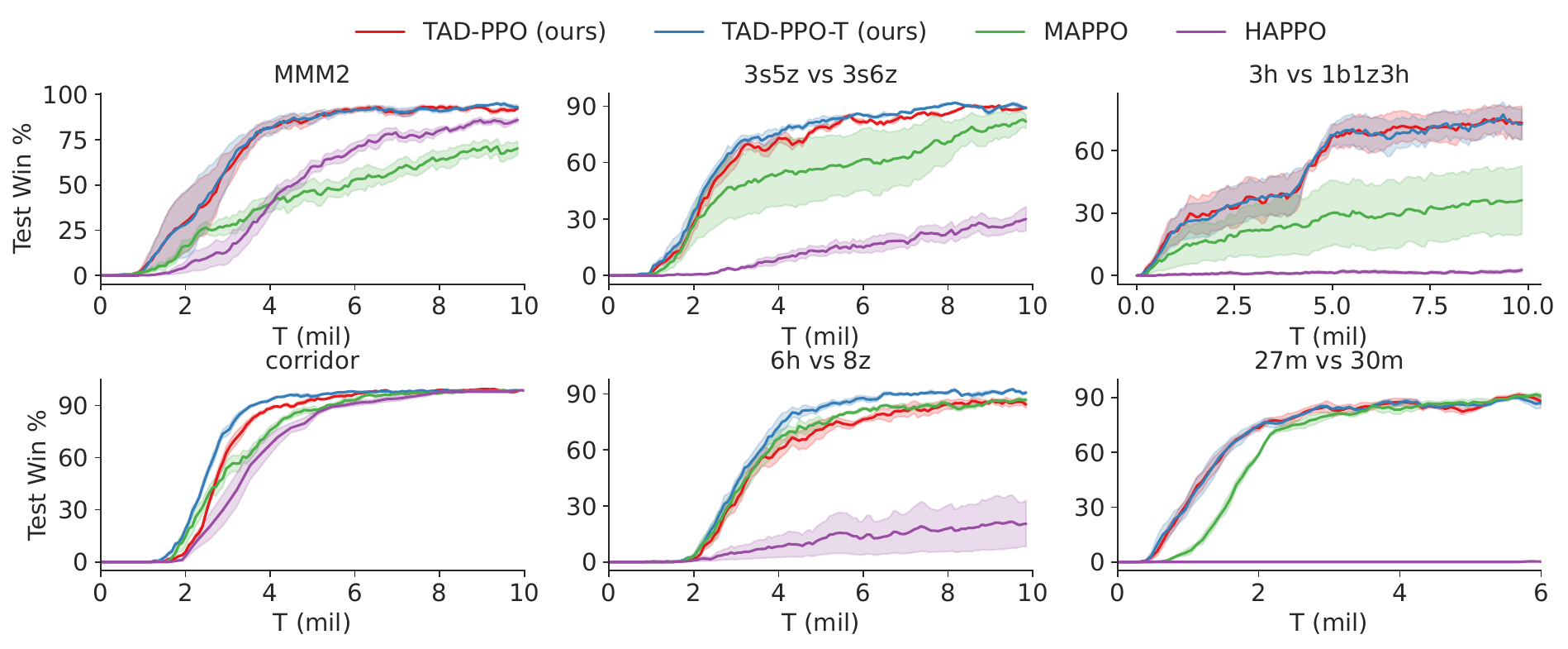}
  \caption{Performance comparison on SMAC benchmark.}
  \label{fig:smac_results}
\end{figure*}


\begin{figure*}[!h]
  \centering
    \includegraphics[width=.95\linewidth]{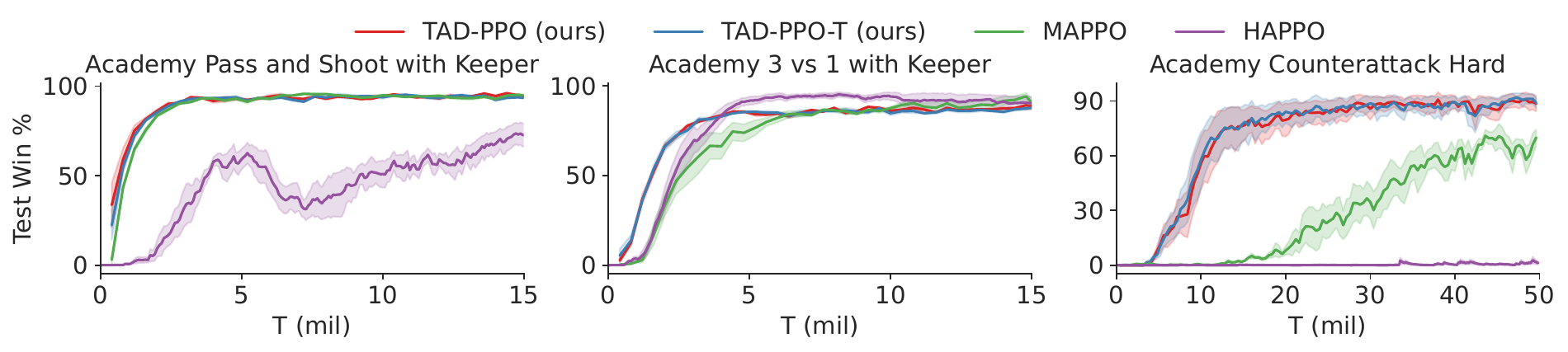}
  \caption{Performance comparison on GRF benchmark.}
  \label{fig:grf_per}
\end{figure*}

We benchmark our method on all 14 maps which are classified as easy, hard, and super hard {on the StarCraft Multi-Agent Challenge (SMAC) benchmark}. {Besides}, we design a novel map {$\mathtt{3h\_vs\_1b1z3h}$} with one local optimal point (game winning strategy) and one global optimal point {in it.}  {We compare our TAD-PPO with MA-PG baselines in this map.} In this section, we illustrate the comparison of performance on five super hard maps  {provided in SMAC benchmark}, $\mathtt{MMM2}$, $\mathtt{3s5z\_vs\_3s6z}$, $\mathtt{corridor}$, $\mathtt{6h\_vs\_8z}$, $\mathtt{27m\_vs\_30m}$, and {a new map ($\mathtt{3h\_vs\_1b1z3h}$) designed by this work} based on the SMAC benchmark. The comparison of performance {for other SMAC maps} is {presented} in \cref{sec:smac_all}. We also benchmark our method on SMAC v2, an updated MARL testbed on the mostly used SMAC benchmark with the inclusion of stochasticity. Details  can be found in \cref{app:exp_smacv2}.

\begin{figure}[H]
  \centering
    \includegraphics[width=.9\linewidth]{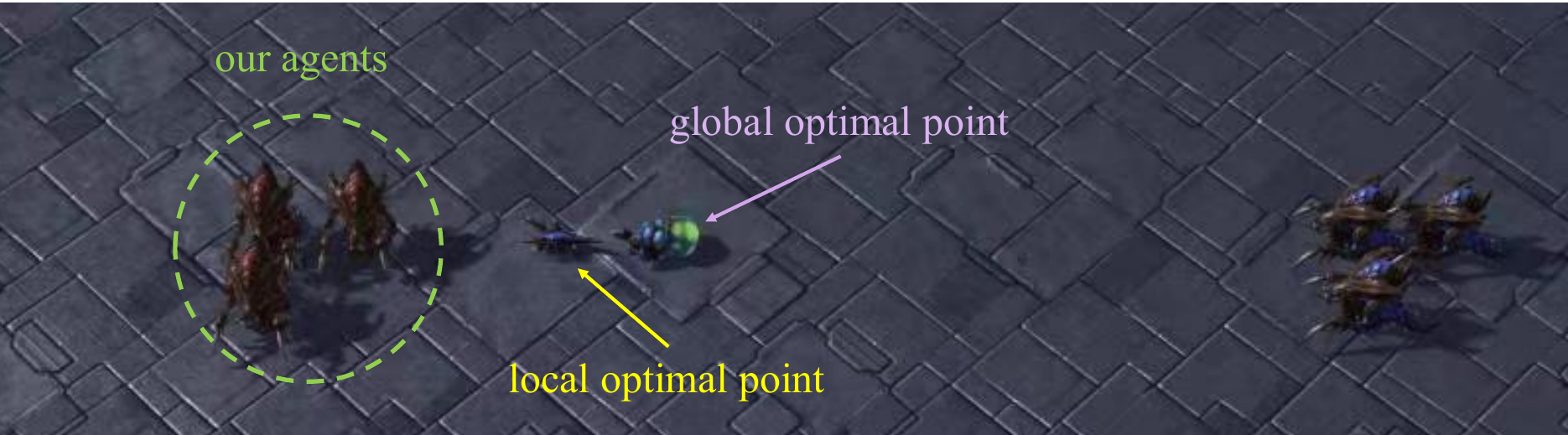}
  \caption{Demonstration of the global optimal and suboptimal solutions in map $\mathtt{3h\_vs\_1b1z3h}$.}
  \label{fig:3h_vs_1b1z3h}
\end{figure}

{} \cref{fig:smac_results} shows that in these super hard SMAC maps, our {TAD-PPO-T (joint policy at the transformation stage), and TAD-PPO (decentralized policies at the distillation stage)} algorithm, based on the TAD framework, demonstrate significant outperformance compared with MAPPO in half maps and outperforms HAPPO in five out of six maps. Super hard maps are typically hard-exploration tasks. {T}aking benefit of {TAD's} global optimality guarantee, our approach can exploit better {under the use of} the same exploration strategy as MAPPO ({which explores} based on the entropy of learned policies).

\cref{fig:smac_results} {validates} our transformation and distillation framework. With {sequentially} transformed information, TAD-PPO-T achieves superior performance { and enjoys the merits of global optimality of its corresponding single-agent PPO algorithm).} {By distillation, TAD-PPO, the counterpart of TAD-PPO-T, can execute in a decentralized manner with comparable performance.}

We further discuss {the} learning performance in our {newly-}designed map $\mathtt{3h\_vs\_1b1z3h}$ {In this task,} we control three Hydralisks, while {the} opponent controls three Hydralisks, one low-damage Zergling, and one Baneling with high area {explosive} damage. The initial positions of our controlled agents and opposing agents are shown in Figure~\ref{fig:3h_vs_1b1z3h}, where our agents are initialized on the left side of the map. The {explosive damage} of the Baneling can only be avoided if all agents gather fire to it instead of {firing against} the {nearer neighbor} Zergling {in the first place
}. {G}athering fire to the Zergling could be a suboptimal equilibrium -- any agent's fire transfer from Zergling to Baneling {does} not influence the explosion of Baneling but reduce the efficiency of killing Zergling. In this case, {TAD} can jump out of the local optimal point after 4M time steps, while the state-of-the-art policy-based MARL algorithms like MAPPO can only fall into the local optimal point {on firing against the Zergling}. This phenomenon emphasizes the superior advantages of our algorithm in the environment with local optimal points.

We conduct further experiments to compare TAD-PPO with more latest MARL algorithms that also consider agents' dependency in their design,  including DFOP\citep{wang2022dfop}, and ACE\citep{li2023ace}. Details are elaborated in \cref{app:exp_smacv2}.

\subsection{Google Research Football}
\label{grf}


In this section, we test our approach against state-of-the-art policy-based MARL baselines on another multi-agent benchmark {} \textit{Google Research Football} (GRF). In GRF academic scenarios, {we control the agents on the left side, and attempt to score against the opposing team's defenders and goalkeeper.} Following {the evaluation settings of} MAPPO \citep{yu2021surprising}, we use sparse rewards with both SCORING and CHECKPOINT for {the evaluation of } our approach and all baselines, {and we use} the official \textit{simple115v2} ~\citep{kurach2019google}, an 115-dimensional vector  as the observation. Each agent has a discrete action space of size 19, including moving in eight directions, sliding, shooting, and passing.

We benchmark our approach on three different GRF scenarios with 2, 3, 4 agents respectively (\cref{fig:grf_per}). In GRF scenarios, agents  {are required to coordinate timing and positions, orchestrating an offense to capitalize fleeting opportunities.} The {coordination} between agents is difficult {} because of the sparsity of {feedback on} agents' crucial movements. {} TAD framework {requires agents to assimilate sequential information from their preceding agents' actions during training, which fosters enhanced coordination for sophisticated cooperation.}  Compared with the state-of-the-art policy-based {baselines of MAPPO and HAPPO}, our approach achieves superior performance in the most challenging scenario \texttt{Academy\allowbreak\_Counterattack\allowbreak\_Hard}
on GRF and achieves comparable performance in other scenarios, {highlighting the generality} of the TAD framework.




\section{Conclusion}
Prior efforts on cooperative MARL algorithms are able to learn decentralized policies with significant performance and theoretical properties. However, achieving optimality in practice remains as a challenge. In this work, we identify the pivotal role of optimization in bridging the gap between theory and practice. This study represents the first theoretical characterization of the suboptimality issue that arises when using gradient descent for optimization in multi-agent tasks. We formally show that MARL algorithms' decentralized structure constitutes a constrained optimization procedure, causing gradient descent to asymptotically get stuck in local optima. To mitigate this, we introduce the \textit{Transformation And Distillation} framework, enabling the use of standard, off-the-shelf single-agent RL methods to address cooperative multi-agent tasks while retaining their global optimality. Building upon the TAD framework, TAD-PPO sets a new benchmark with optimality guarantees, surpassing existing baselines across various tasks.

This paper offers two key insights to the community: Firstly, the empirical design of MARL algorithms should incorporate specific optimization techniques; and secondly, the TAD framework acts as a bridge, facilitating the smooth application of single-agent algorithms to tackle multi-agent problems with well-established guarantees.
Limitations and future avenues of this work are discussed in detail at \cref{app:limitations_and_discussions}.



\bibliographystyle{ACM-Reference-Format} 
\bibliography{references}


\newpage
\appendix
\onecolumn
\section{Omitted Definitions in Section  \ref{sec:preliminary}}

\subsection{RL Models}
\label{app:RLmodels}
In single-agent RL (SARL), an agent interacts with a Markov Decision Process (MDP) to maximize its cumulative reward \citep{sutton2018reinforcement}. An MDP is defined as a tuple $(\mathcal S,\mathcal A,P,r,\gamma)$, where $\mathcal S$ and $\mathcal A$ denote the state space and action space. At each time step $t$, the agent observes the state $s_t$ and chooses an action $a_t\in \mathcal A$, where $a_t\sim \pi(s_t)$ depends on $s_t$ and its policy $\pi$. After that, the agent will gain an instant reward $r_t=r(s_t,a_t)$, and transit to the next state $s_{t+1}\sim P(\cdot|s_t,a_t)$. $\gamma$ is the discount factor. The goal of an SARL agent is to optimize a policy $\pi$ that maximizes the expected cumulative reward, i.e., $\mathcal{J}(\pi)=\mathbb{E}_{s_{t+1}\sim P(\cdot{|s_t,\pi(s_t)})}\left[\sum_{t=0}^\infty\gamma^tr(s_t,\pi(s_t))\right]$.

we also present a framework of Multi-agent MDPs (MMDP) \citep{boutilier1996planning}, a special case of Dec-POMDP, to model cooperative multi-agent decision-making tasks with full observations. MMDP is defined as a tuple $\langle \mathcal{S}, \mathcal{A}, P, r, \gamma,n\rangle$, where  $\mathcal{S}$, $\mathcal{A}$, $P$, $r$, $\gamma$ and $n$ are state space, individual action space, transition function, reward function, discount factor, and the number of agents in Dec-POMDP, respectively. Due to the full observations, at each time step, the current state $s$ is observable to each agent. For each agent $i$, an individual policy $\hat{\pi}_i(a|s)$ represents a distribution over actions conditioned on the state $s$. Agents aim to find a joint policy $\hat{\bm{\pi}}=\langle \hat{\pi}_1, \dots, \hat{\pi}_n\rangle$ that maximizes a joint value function $\widehat{V}^{\hat{\bm{\pi}}}(s)$, where denoting $\widehat{V}^{\hat{\bm{\pi}}}(s)=\mathbb{E}\left[\sum_{t=0}^{\infty} \gamma^t r_t|s_0=s,\hat{\bm{\pi}}\right]$.

\section{Omitted Proofs in Section  \ref{sec:motivation}}
\label{app:motivation}

\subsection{Some Lemmas for Gradient Descent and Local minima}

Since we are going to analyze the behavior of gradient descent (GD), we introduce some preliminaries and lemmas related to GD here. 

GD in this paper is chosen to minimize some real-valued function $f(x)$. It starts from some initial point $x_0$, and updates the points as $x_{t+1} = x_{t}  -\alpha\nabla f(x_t)$, where $\alpha>0$ is a sufficiently small constant learning rate.
Classical analysis of GD requires the objective function $f(x)$ to be smooth and $\alpha$ to be a sufficiently small (or diminishing) learning rate to guarantee convergence to a local optimum.

The $L$-smoothness assumption is usually adopted, which says the derivative function is $L$-Lipschitz.

\begin{definition}[$L$-smoothness]

A function $f(x)$ is said to be $L$-smooth, if its derivative function is $L$-Lipschitz, i.e., $\forall x,x':\|\nabla f(x)-\nabla f(x')\|\le L\|x-x'\|$.
\end{definition}

We then present several lemmas for $L$-smooth functions (\cref{app-lemma:lsmooth1,app-lemma:smooth-gd-monotone,app-lemma:smooth-gd-converge}), which are classical results in the analysis of gradient descent.

\begin{lemma}
\label{app-lemma:lsmooth1}
  Suppose $f$ is $L$-smooth, then $\forall x,x_0:f(x)\le f(x_0)+\langle\nabla f(x_0),x-x_0\rangle+\frac L 2\|x-x_0\|^2$.

\end{lemma}

\begin{proof}
  Let $g(t) = f(x_0+t(x-x_0))$, then $g(1) = f(x), g(0) = f(x_0)$.
  \begin{align*}
    f(x)-f(x_0)& = g(1)-g(0) = \int_0^1 g'(t) \dd t\\
    & = \int_0^1\langle \nabla f(x_0+t(x-x_0)),x-x_0\rangle\dd t\\
    & = \int_0^1\langle \nabla f(x_0+t(x-x_0))-\nabla f(x_0),x-x_0\rangle \dd t + \langle \nabla f(x_0),x-x_0\rangle\\
    & \le \int_0^1\|\nabla f(x_0+t(x-x_0))-\nabla f(x_0)\|\cdot \|x-x_0\|\dd t + \langle  \nabla f(x_0),x-x_0\rangle \\
    & \le \int_0^1 Lt\|x-x_0\|\cdot\|x-x_0\|\dd t+\langle \nabla f(x_0),x-x_0\rangle\\
    & = \frac L 2\|x-x_0\|^2+\langle \nabla f(x_0),x-x_0\rangle
  \end{align*}
\end{proof}

Then we have the following lemma for GD.

\begin{lemma}
\label{app-lemma:smooth-gd-monotone}
  Under the above conditions, if $\alpha\in \left(0,\frac{2}{L}\right)$, then $f(x_{t+1})\le f(x_t) - c\|\nabla f(x_t)\|^2$ for some constant $c > 0$.
\end{lemma}

\begin{proof}
    \begin{align*}
      f(x_{t+1})&\leq f(x_{t})+\langle \nabla f(x_{t}),x_{t+1}-x_{t} \rangle +\frac{L}{2}\|x_{t+1}-x_{t}\|^{2}  \tag{By \cref{app-lemma:lsmooth1}}\\
        &= f(x_{t})+\langle \nabla f(x_{t}),-\alpha\nabla f(x_t) \rangle +\frac{L}{2}\alpha^{2}\|\nabla f(x_t)\|^{2} \\
        &\leq f(x_{t})-\alpha \|\nabla f(x_{t})\|^{2} + \frac{L}{2} \alpha^2 \|\nabla f(x_{t})\|^{2} \\
        &=f(x_{t})-\alpha\left(1-\frac{L}{2} \alpha\right) \|\nabla f(x_{t})\|^{2}
    \end{align*}

    Letting $c = \alpha\left(1-\frac{L}{2} \alpha\right) > 0$ completes the proof.

\end{proof}

\begin{lemma}
\label{app-lemma:smooth-gd-converge}
    Under the above conditions, if $\alpha\in \left(0,\frac{2}{L}\right)$, then either $f(x)$ decreases to $-\infty$ or $f(x)$ converges and $\lim_{t\rightarrow +\infty} \nabla f(x_t) = 0$.
    If $\{x_t\}$ converges, the limiting point is a stationary point of $f$.
\end{lemma}

\begin{proof}
    The first part of this lemma is a corollary of \cref{app-lemma:smooth-gd-monotone} and Monotone Convergence Theorem.
    The second part can be proved by contradiction. Supposing the limit point $x^*$ is not a stationary point, then $\nabla f(x^*)$ is non-zero due to the Lipschitz continuity, which derives a contradiction by taking one-step gradient update over a sufficiently close $x_t$ near $x^*$. 
\end{proof}

Besides, the following lemma for local minima is simple but crucial in our analysis.

\begin{lemma}
    \label{app-lemma:local-mini-continuous-local}
    If $g$ is a continuous function, and $y_0 = g(x_0)$ is a local minimum of the function $f(y)$, then $x_0$ is a local minimum of $f\circ g$.
\end{lemma}

\begin{proof}
    Since $y_0$ is a local minimum of $f(y)$, there exists a small open set $O\ni y_0$, such that $\forall y\in O : f(y)\ge f(y_0)$.

    Due to the continuity of $g$, $O_1 = g^{-1}(O)$ is also open. So that we can find a small neighborhood $O_2$ of $x_0$, such that $x_0\in O_2\subseteq O_1$, which implies $g(O_2)\subseteq O$.
    Then we have $\forall x\in O_2: f(g(x))\ge f(y_0)=f(g(x_0))$.
\end{proof}


\subsection{The Suboptimality Theorem for Multi-agent Policy Gradient Algorithm (MA-PG)}
\label{app:proofs-mapg}

Here we present a proof of \cref{thm:ma-pg-fails}.

Recall the policy factorization and the loss function of MA-PG:
\begin{align*}
&\forall s\in \mathcal S:\quad \bm{\pi}(\bm{a}|s;\Theta) = \prod_{i=1}^n \pi_i(a_i|s;\Theta)     \tag{\ref{eq:pi-factor}}\\
&\mathcal L(\Theta) = -\mathcal J(\bm \pi_\Theta)   \tag{\ref{eq:mapg-loss}}
\end{align*}

To simplify our analysis, we abuse the notation a little bit to rewrite the policy factorization as
\begin{align}
\forall s\in \mathcal S:\quad \bm{\pi}(s;\Theta) = \bigotimes_{i=1}^n \pi_i(s;\Theta)
\end{align}

where $\pi_i(s;\Theta)$ is now a probability vector in the probability simplex $\Delta(\mathcal A)$, and $\bm\pi(s;\Theta)$ is the tensor product of these vectors.

We also use $L(\bm\pi) = -\mathcal J(\bm \pi)$ or $L(\pi_1,\cdots,\pi_n) = -\mathcal J(\bigotimes\pi_i) $ to represent the loss of a policy, which is essentially the same as $\mathcal L$ except for its domain. 

Our analysis is based on the implicit assumptions:

\begin{assumption}[smoothness]
\label{assump:mapg-param-smooth}
For all $s\in \mathcal S$, the function $\bm\pi(s;\Theta)$ is twice differentiable.
\end{assumption}

\begin{assumption}[bounded gradient]
\label{assump:mapg-param-boundedgrad}
For all $s \in \mathcal S,i\in\{1,\cdots,n\}$, $\|\nabla \pi_i(s;\Theta)\|$ and $\|\nabla^2 \pi_i(s;\Theta)\|$ are bounded.
\end{assumption}

\begin{assumption}[completeness]
\label{assump:mapg-param-complete}
For all $s\in \mathcal S$, all probability distribution in the product space $\bigtimes_{i=1}^n\Delta(\mathcal A)$ is representable by $\bm \pi(s;\Theta)$.
\end{assumption}

The smoothness assumption is common. The bounded gradient assumption is also wildly used in analyzing GD and non-convex optimization problems\citep{Recht2011HogwildALboundedgrad, Nemirovski2008RobustSAboundedgrad, ShalevShwartz2007PegasosPEboundedgrad}. The completeness assumption (or regularity assumption) is made to prevent degenerated cases such as the function being a constant ($\Theta$ does not affect the policy), which is of little significance to analyze.

\begin{reptheorem}{thm:ma-pg-fails}
There are tasks such that, when the parameter $\Theta$ is initialized in certain region $S$ with positive volume, MA-PG (\cref{eq:pi-factor,eq:mapg-loss}) converges to a suboptimal policy with a small enough learning rate $\alpha$.
\end{reptheorem}

\begin{proof}
    We prove the theorem in three steps.
    First, we construct a matrix game just like \cref{matgame1}, which is a 1-step MMDP.
    Second, we prove the loss function in this task contains $\Omega(|\mathcal A|)$ suboptimal local minima.
    Finally, we prove that for each local minima we construct, there is a neighborhood $S$, such that the gradient descent process will never escape from it with a small enough learning rate.
    Moreover, the policy will converge.

    We construct a 2-agent matrix game here with the payoff matrix $M$:
    
    $$M = \left(\begin{array}{cccc}
        1 & 0 & \cdots & 0  \\
         0 & 2 &\cdots & 0\\
         \vdots& \vdots&\ddots &\vdots\\
         0 & 0 & \cdots & |\mathcal A|
    \end{array}\right)$$
    
    Since there is only one state, we omit the state $s$ in our policy.
    
    We claim that for any deterministic policy selecting the entry of the diagonal of $M$, the corresponding $\Theta$ is a local minimum of $\mathcal L$, where the existence of $\Theta$ is guaranteed by \cref{assump:mapg-param-complete}.
    
    That is, we have $\Theta_1,\cdots, \Theta_{|\mathcal A|}$ such that $\bm\pi(\Theta_i)$ is the deterministic policy selecting the entry $M_{i,i}$. And if the claim is true, $\Theta_i$ will be a suboptimal local minimum of $\mathcal L$ for $i=1,\cdots,|\mathcal A|-1$.
    
    We now prove that for any $i\in\{1,\cdots, |\mathcal A|\}$, $\Theta_i$ is a local minimum.
    
    Fix some $i$. Denote $\bm p=\bm q=\delta_i$, where $\delta_i\in \Delta(|\mathcal A|)$ is the one-hot probability vector indicating the deterministic policy of selecting the $i$-th action. By definition of $\Theta_i$, we have $\bm\pi(\Theta_i) = \bm p\otimes \bm q$.
    
    The payoff of joint policy $\bm\pi(\Theta_i)$ is thus $\mathcal J(\bm\pi(\Theta_i)) = \bm p^\top M \bm q$. 
    Now we investigate some policy $\bm p'\otimes \bm q'$ in a sufficiently small neighborhood of $\bm p\otimes \bm q$ in the product space. 
     
    Denote $\Delta \bm p = (\bm p - \bm p'), \Delta\bm q = (\bm q-\bm q'), \epsilon_1 = \min\{\frac 1 2\|\Delta\bm p\|_1, \frac 1 2\|\Delta \bm q\|_1\}, \epsilon_2 = \max\{\frac 1 2\|\Delta\bm p\|_1, \frac 1 2\|\Delta \bm q\|_1\}$.
    We have $\epsilon_2\ge \Delta\bm p(i), \Delta\bm q(i)\ge \epsilon_1, -\epsilon_1\ge \Delta\bm p(j),\Delta\bm q(j)\ge -\epsilon_2$, for $j\neq i$.
    
    Therefore,
    \begin{align}
      \bm p'^\top M\bm q' &=  (\bm p-\Delta \bm p)^\top  M(\bm q-\Delta \bm q)   
      \nonumber\\
      &= \mathcal J(\bm p, \bm q) - \Delta \bm p^\top M\bm q - \bm p^\top M \Delta \bm q + \Delta \bm p^\top M\Delta \bm q 
      \nonumber\\
      &= \mathcal J(\bm p, \bm q) - \delta_i^\top M (\Delta\bm p + \Delta\bm q) + \Delta \bm p^\top M\Delta \bm q 
      \nonumber\\
      &= \mathcal J(\bm p, \bm q) - i (\Delta \bm p(i) + \Delta \bm q(i)) + \sum_{k=1}^{|\mathcal A|} k \Delta \bm p(k)\Delta \bm q(k)
      \nonumber\\
      &\le \mathcal J(\bm p, \bm q) - (\epsilon_1+\epsilon_2)  + \epsilon_2^2   |\mathcal A|^2
      \nonumber\\
      &= \mathcal J(\bm p, \bm q) - \Omega(\epsilon_2) + O(\epsilon_2^2) 
      \\
      &< \mathcal J(\bm p, \bm q) \tag*{Let $\epsilon_2$ be sufficiently small}
      \nonumber
    \end{align}
    
    This inequality suggests that $\delta_i\otimes \delta_i$ is a local minimum of the payoff in the space $\Delta(|A|)^2$.
    
    Then due to the continuity of $\bm\pi(\Theta)$ (\cref{assump:mapg-param-smooth}), $\Theta_i$ would be a local minimum of $\mathcal L(\Theta)$ in the space $\text{dom}(\Theta)$ by \cref{app-lemma:local-mini-continuous-local}.
    
    We now prove the for all $i=1,\cdots,|\mathcal A|-1$, there is a neighborhood $S$ of $\Theta_i$, such that for all $\Theta\in S$, a gradient update $\Theta'\gets \Theta-\alpha\nabla \mathcal L(\Theta)$ still lies in $S$ for small enough $\alpha$.
    
    In general, the local optimality does not imply local convexity\footnote{Consider a $C^\infty$ function $f(x) = \exp(-1/x^2)(\sin(1/x)+1)$ at $x=0$.}, and therefore no local convergence result is guaranteed for GD by initializing near a local optimum.
    However in this case, our function has a higher-level structure in the space $\Delta(|\mathcal A|)^2$, which is essential for us to establish our result.
    
    We fix some $i\in \{1,\cdots,|\mathcal A|-1\}$.
    
    Denote $p(\Theta) = ( \pi_1(\Theta) ,\pi_2(\Theta)) \in \Delta(|\mathcal A|)^2$, and let $x = p(\Theta_i)$.\footnote{Note that the range of $p$ is $\Delta(|\mathcal A|)^2 = \Delta(|\mathcal A|)\times \Delta(|\mathcal A|)$, while the range of $\bm \pi$ is $\Delta (|\mathcal A|^2) $.}
    Let $\epsilon>0$ be a sufficiently small constant to be determined later.
    In the space $\Delta(|\mathcal A|)^2$, we consider an open ball $B_{\epsilon}(x)$ centered at $x$ with radius $\epsilon$, where the distance is the total variation $\|\cdot\|_{TV}$, i.e., half-$L^1$-norm $\frac 1 2 \|\cdot\|_1$.
    The selection of the norm is not crucial, since all norms on a finite-dimensional linear space are equivalent\footnote{This is a classical result in finite linear space: $\forall \|\cdot\|_1,\|\cdot\|_2\exists c_1,c_2>0:c_1\|\cdot\|_1\le \|\cdot\|_2\le c_2\|\cdot\|_1$.}.

    We let $R = p^{-1}(B_{\epsilon}(x)) \ni \Theta_i$.
    And for any $\theta\in R$, denote the gradient update at $\theta$ as $s_\theta = -\alpha\nabla\mathcal L(\theta)$. 
    Adopting the notations above, we prove the following lemma.

    \begin{lemma}
    \label{app-lemma:MAPG-onestep-L-monotone}
        $\mathcal L(\theta + s_\theta)\le \mathcal L(\theta)$ for sufficiently small $\alpha$.
    \end{lemma}

    \begin{proof}
        According to \cref{assump:mapg-param-boundedgrad}, we have $\|\nabla p (\Theta)\|$ and $\|\nabla^2 p (\Theta)\|$ are bounded by some constant $G$.
        Then for any $\Theta$, by applying the chain rule and the product rule of gradient, we have

        \begin{align*}
        \|\nabla^2 \mathcal L(\Theta)\| &= \|\nabla^2 p(\Theta) \nabla L(x) + \nabla p(\Theta) \nabla^2 L(x)\| \\
        &\le \|\nabla ^2 p(\Theta)\|\|\nabla L(x)\| + \|\nabla p(\Theta)\|\|\nabla ^2 L(x)\|\\
        &\le G |\mathcal A| + G\sqrt {|\mathcal A|}
        \end{align*}

        Let $H = G |\mathcal A| + G\sqrt {|\mathcal A|}$, then $\nabla \mathcal L$ is $H$-Lipschitz by the mean-value theorem, which means $\mathcal L$ is $H$-smooth. 
        
        Then by \cref{app-lemma:smooth-gd-monotone}, $\mathcal L(\theta + s_\theta)$ is not increasing by making $\alpha<2/H$.

    \end{proof}

    We further denote $S_\delta = \{y\in \Delta(|\mathcal A|)^2: L(y)\le L(x) + \delta\} \cap B_{\epsilon}(x)$.
    Here comes the key lemma.

    \begin{lemma}
        \label{app-lemma:MAPG-local-step}
        There exist $\delta,\epsilon',\epsilon$, such that
        \begin{enumerate}
            \item $\delta>0$ and $0<\epsilon'<\epsilon$.
            \item $S_\delta \subset B_{\epsilon'}(x)\subset B_\epsilon (x)$.
        \end{enumerate}
    \end{lemma}
    \begin{proof}
        Recall that $x = p(\Theta_i) = (\delta_i, \delta_i)$. Let $y=(\bm p, \bm q)\in B_\epsilon(x)$, solve the inequality $L(y) = -  \bm p^\top M \bm q  \le L(x)+\delta $, we have,\footnote{$\delta$ here is a small constant, while $\delta_i$ is the one-hot vector defined before.}
        \begin{align}
        - \left(i  (1- \|\delta_i - \bm p\|) (1- \|\delta_i - \bm q\|) + \sum_{(u,v)\neq (i,i)}M(u,v) \bm p(u)\bm q(v) \right) & \le - i + \delta 
        \nonumber\\
        i \|x-y\| - i \|\delta_i - \bm p\| \|\delta_i - \bm q\| -\sum_{(u,v)\neq (i,i)}M(u,v) \bm p(u)\bm q(v) &\le  \delta 
        \nonumber
        \end{align}
        \begin{align}
        \|x-y\|&\le \frac 1 i\left(\delta + |\mathcal A|^3 \|x-y\|^2\right)
        \nonumber \\
        \|x-y\|&\le \frac 1 i\left(\delta + \epsilon^2 |\mathcal A|^3\right)
        \label{ineq:mapg-ball-ball-levelset}
        \end{align}
        
        We choose $\epsilon$ and $\delta$ small enough such that $0< \delta < i\epsilon/3-\epsilon^2|\mathcal A|^3$, then we will have $\|x-y\|<\epsilon / 2$ by plugging it into \cref{ineq:mapg-ball-ball-levelset}.
        This suggests that $S_{\delta} \subset B_{\epsilon / 2} (x) \subset B_\epsilon (x)$ and completes our proof.
    \end{proof}

    By \cref{app-lemma:MAPG-local-step}, we can find some $\delta>0$ and $0<\epsilon'<\epsilon$, such that $S_\delta \subset B_{\epsilon'}(x)\subset B_\epsilon (x)$. Then for any $\theta\in p^{-1}(S_\delta)$, we have
    
    \begin{enumerate}    
        \item $p(\theta+s_\theta) \in \{y\in \Delta(|\mathcal A|)^2: L(y)\le L(x) + \delta\}$, because $\mathcal L(\theta + s_\theta)$ is not increasing (\cref{app-lemma:MAPG-onestep-L-monotone});
        \item $p(\theta+s_\theta)\in B_\epsilon (x)$ by letting $\alpha$ small enough such that $\alpha <\frac 1 {2G^2} (\epsilon - \epsilon' )$, where $G$ is an upper bound of $\|\nabla p(\theta)\|$ (\cref{assump:mapg-param-boundedgrad}).
    \end{enumerate}

    The second bullet holds because

    \begin{align*}
        \|p(\theta+ s_\theta) - p(\theta)\| & \le G \|s_\theta\| \tag{\cref{assump:mapg-param-boundedgrad} and mean-value theorem}\\
        &= \alpha G \|\nabla p(\theta)\|\\
        &\le \alpha G^2 \tag{\cref{assump:mapg-param-boundedgrad}}\\
        &=(\epsilon-\epsilon')/2
    \end{align*}
    
    These give us $\forall \theta\in p^{-1}(S_\delta):\theta + s_\theta \in p^{-1}(S_\delta)$.
    Again by the local $H$-smoothness and \cref{app-lemma:smooth-gd-converge}, we are able to claim that $\mathcal L(\theta)$ converges to a suboptimal value when $\theta$ is initialized in $S_\delta$.

\end{proof}
{In the above proof, we conduct asymptotic analysis.} Specifically, we demonstrate a parameter space region S, lacking the global optimum. Here, a gradient step $\theta \leftarrow \theta-\alpha \nabla \mathcal{L}(\theta)$ still lies in S (\cref{app-lemma:MAPG-local-step}), and reduces the loss (\cref{app-lemma:MAPG-onestep-L-monotone}), which makes it converge in a sub-optimal local optimum by the monotonicity argument mirroring the classical analysis in SGD with a smooth function (\cref{app-lemma:smooth-gd-converge}, also cf. \citet{bottou2018optimization}).


\subsection{The Suboptimality Theorems for Value-decomposition Algorithms (VD)}
\label{app:proofs-vd}


We are going to present three suboptimality theorems here for VD progressively.
To start with, we recall the definition of VD (\cref{eq:value-decomp,eq:vd-loss,eq:igm}).

\begin{align*}
    \text{Value-decomposition: }&Q(s, \bm a;\Theta)  = f_{\text{mix}}(Q_1(s,\cdot), \cdots,Q_n(s,\cdot),s ,\bm a;\Theta)\tag{\ref{eq:value-decomp}}\\
    \text{Loss function: }&\mathcal L(\Theta) = \frac 1 2\mathbb E_{s,\bm a}\left[
     Q(s, \bm a;\Theta) - (\mathcal T Q_{\Theta})(s,\bm a)
    \right]^2 \tag{\ref{eq:vd-loss}}\\
    \text{IGM condition: }&\forall s:\quad \mathop{\arg\max}_{\bm{a}\in\bm{\mathcal{A}}} Q(s, \bm{a}) \supseteq \bigtimes_{i=1}^n\mathop{\arg\max}_{a_i\in\mathcal{A}} Q_i(s, a_i)\tag{\ref{eq:igm}}
\end{align*}

To reduce unnecessary difficulties in understanding, we explicitly describe the joint Q-value as a function $f_{\text{mix}}$ of individual Q-values in \cref{eq:value-decomp}, which is the origin of the name ``value-decomposition".
However, one may notice that, since the local Q-functions $Q_1,\cdots,Q_n$ are also parameterized by $\Theta$, then it is unnecessary for the joint Q-function to take local Q-values as inputs.
That is, \cref{eq:value-decomp} is redundant.
The only thing we need is to have $n+1$ functions satisfying \cref{eq:igm}, where one is the joint Q-function $Q(s,\bm a;\Theta)$, and the other $n$ are the local Q-functions $Q_1(s,a_1;\Theta),\cdots,Q_n(s,a_n;\Theta)$, all parameterized by $\Theta$.

So that our definition of VD goes beyond the original form of ``value-decomposition".

To facilitate our analysis, we first introduce two kinds of completeness here.

\begin{definition}[Completeness of $Q$-function Class]
\label{def:complete-igm-complete}
Consider a value-decomposition algorithm $D$ (\cref{eq:igm,eq:vd-loss}).
We denote $\mathcal S,\mathcal A,n$ as the state space, individual action space and the number of agents here.

$D$ is said to have a \textbf{complete} Q-function class if for any $ f\in \mathbb R^{\mathcal S\times \mathcal A}$, there is a $ \Theta\in \text{dom}(\Theta)$, such that $\forall s\in \mathcal S,\bm a\in\mathcal A^n :Q(s,\bm a;\Theta)=f(s,\bm a)$.

Further more, $D$ is said to have an \textbf{IGM-complete} Q-function class if for any $ f\in \mathbb R^{\mathcal S\times \mathcal A}$, and for any deterministic joint policy $\bm \pi:\mathcal S\rightarrow \mathcal A^n$ such that $\forall s\in \mathcal S:\bm \pi(s)\in\argmax f(s,\cdot)$, there is a $ \Theta\in \text{dom}(\Theta)$, such that 

\begin{enumerate}
    \item $\forall s\in \mathcal S,\bm a\in\mathcal A^n :Q(s,\bm a;\Theta) = f(s,\bm a)$.
    \item $\forall i\in[n],s\in \mathcal S:  \argmax Q_i(s,\cdot ;\Theta)=\{\pi_i(s) \}$, where $\bm\pi =(\pi_1,\cdots,\pi_n)$.
\end{enumerate}

\end{definition}

In a nutshell, the \textbf{completeness} says we can factorize any joint Q-function in $\mathbb R^{\mathcal S\times \mathcal A}$ into local Q-functions $Q_1,\cdots, Q_n$ by some parameter $\Theta$, while the \textbf{IGM-completeness} says that when factorizing some joint Q-function with multiple greedy policies (i.e., $\exists s:|\argmax Q(s,\cdot)|>1$), we can select any one of them such that after factorization, the greedy policies of local Q-functions are consistent to the selected greedy policy.\footnote{The IGM-completeness is a mild and natural restriction of the completeness when the IGM condition is satisfied. QPLEX\citep{Wang2021QPLEXDD} satisfies the IGM-completeness.}

Then we abuse notations similar to the case of MA-PG (\cref{app:proofs-mapg}).

For all $i\in \{1,\cdots, n\}$, we denote $Q_i(\Theta) \in \mathbb R^{\mathcal S\times \mathcal A}$, where $Q_i(\Theta) (s,a) = Q_i(s,a;\Theta)$.

Similarly, denoting $Q(\Theta)\in \mathbb R^{\mathcal S\times \mathcal A^n}$, where $Q(\Theta)(s,\bm a) = Q(s,\bm a;\Theta)$.

We also use $L(f) = \frac 1 2\mathbb E_{s,\bm a}\left[ f(s,\bm a) - (\mathcal T f)(s,\bm a) \right]^2$, where $f\in \mathbb R^{\mathcal S\times \mathcal A^n}$ and $\mathcal T$ is the Bellman operator. $L$ is essentially the same as $\mathcal L$ except its domain.

Before we prove \cref{thm:vd-fails}, we introduce two more versions of the suboptimality theorems for VD (\cref{thm:igm-complete-vd-fails,thm:complete-compact-vd-fails}).


\begin{theorem}
\label{thm:igm-complete-vd-fails}
There are tasks such that, when the parameter $\Theta$ is initialized in certain set $S$, any VD with an \textbf{IGM-complete} function class converges to a suboptimal policy.
\end{theorem}

\begin{proof}
    Our proof contains 4 parts.
    
    We will first construct an $n$-player matrix game, which is a $1$-step MMDP.
    Second, we will find a series of functions $f^{(j)}\in \mathbb R^{\mathcal S\times \mathcal A^n}$, such that for all $j$, $f^{(j)}$ is a local optimum of $L$ in a restricted space $R$.
    Third, we decompose $f^{(j)}$ by the value decomposition scheme to find a $\Theta^{(j)}$ such that $Q(\Theta^{(j)}) = f^{(j)}$.
    Finally, we prove that $\Theta^{(j)}$ is a local minimum of $\mathcal L$ by \cref{app-lemma:local-mini-continuous-local} after showing that the function $Q:\text{dom}(\Theta)\rightarrow \mathbb R^{\mathcal S\times \mathcal A^n}$ is a locally a function with range in $R$ (i.e., $Q:\text{dom}(\Theta)\rightarrow R$ near $\Theta^{(j)}$).

    We arbitrarily choose the one-step payoff function (the $n$-dimensional matrix in the $n$-player matrix game) $T\in \mathbb R^{|\mathcal A|^n}$ with different values on all entries.
    Since this is a $1$-step game, we omit the state in our formulation.
    We also have $\mathcal T f = T$ for any value function $f\in \mathbb R^{|\mathcal A|^n}$ since the task lasts only for $1$ step. 
    And the loss function can be simplified as follows.

    $$\mathcal L(\Theta) = \frac 1 2\mathbb E_{\bm a\sim \mathcal D}\left[Q(\bm a;\Theta) - T(\bm a)\right]^2$$

    for some distribution $\mathcal D$ fully supported on $ \mathcal A^n$. 
    
    Let $\bm a^*$ be any joint action in $\mathcal A^n$.
    Denote the restricted value function space $\mathcal R(\bm a^*) = \{ f\in \mathbb R^{|\mathcal A|^n}:\bm a^* \in \argmax  f(\cdot)\}$ as the set of matrices with $\bm a^*$ as its greedy action (may not be unique).

    We present the following lemma.

    \begin{lemma}
        \label{app-lemma:restrict-argmin-unique}
        The set $\argmin_{f\in \mathcal R(\bm a^*)} L(f)$ contains exactly one element.
    \end{lemma}

    \begin{proof}
        We first note that $\mathcal R(\bm a^*)$ is convex in that for any $f_1,f_2\in \mathcal R(\bm a^*)$ and $\beta \in [0,1]$, we have

        \begin{align*}
            \forall \bm a:\quad(\beta f_1 + (1-\beta)f_2)(\bm a^*) &= \beta f_1(\bm a^*) + (1-\beta ) f_2(\bm a^*)\\
            &\ge \beta f_1(\bm a) + (1-\beta ) f_2(\bm a)\\
            & = (\beta f_1 + (1-\beta) f_2)(\bm a)
        \end{align*}

        so that $(\beta f_1 + (1-\beta)f_2) \in  \mathcal R(\bm a^*)$.

        Denote $S = \argmin_{f\in \mathcal R(\bm a^*)} L(f)$.
        It is obvious that $S$ is non-empty.
        Suppose $f_1,f_2$ are two different element in $S$, we investigate the loss of $f = (f_1+f_2)/2$, which is also in $\mathcal R(\bm a^*)$ by the convexity.

        Since $f_1\neq f_2$, there exists some $\tilde{\bm a}$ such that $f_1(\tilde{\bm a}) \neq f_2(\tilde{\bm a })$.
        Denote $l(x,c) = \frac 1 2(x-c)^2$, which is a strict convex function for any fixed $c$, we have

        \begin{align*}
            L(f)&= \frac 1 2\mathbb E_{\bm a\in\mathcal D}[f({\bm a}) - T({\bm a})]^2\\
            & =  \sum_{\bm a}\mathcal D(\bm a) l(f({\bm a}),T({\bm a}) )\\
            & < \sum_{\bm a}\mathcal D(\bm a) \left( \frac 1 2 l(f_1({\bm a}),T({\bm a}))  + \frac 1 2 l(f_2({\bm a}),T({\bm a}) )\right)  \\
            & = \frac 1 2(L(f_1) + L(f_2))\\
            &=L(f_1)
        \end{align*}

        where the inequality is by Jensen's inequality, the strict inequivalence holds since $f_1(\tilde {\bm a }) \neq f_2(\tilde {\bm a})$.
        This contradicts the fact that $f_1$ minimizes $L$ in $\mathcal R(\bm a^*)$.
    \end{proof}

    By \cref{app-lemma:restrict-argmin-unique}, the minimizer is unique.
    Therefore we denote $f_{\bm a^*}$ be the only element in $\argmin_{f\in \mathcal R(\bm a^*)} L(f)$.
    
    It is obvious that $f_{\bm a^*}$ is not the globally optimal solution for $L$ in $\mathbb R^{|\mathcal A|^n}$ unless $\bm a^* \in \argmax T(\cdot)$ since the globally optimal solution is exactly $T$ itself.

    Now by IGM-completeness, we can find a $\Theta_{\bm a^*}$ such that $Q(\Theta_{\bm a^*}) = f_{\bm a^*}$, and $\argmax Q_i(\Theta_{\bm a^*})(\cdot) = \{\bm a^*(i)\}$, for $i=1,\cdots,n$.

    Since $\argmax Q_i(\Theta_{\bm a^*})(\cdot) = \{\bm a^*(i)\}$, we know that the largest action value in $Q_i(\Theta_{\bm a^*})(\cdot)$ is strictly larger than the second largest action value.
    Therefore by the continuity of $Q_i$, there is a small neighborhood $O\ni \Theta_{\bm a^*}$, such that $\forall \Theta\in O: \argmax Q_i(\Theta)(\cdot) = \{\bm a^*(i)\}$ for all $i\in \{1,\cdots,n\}$, and the $\bm a^*\in \argmax Q(\Theta)$ by IGM (\cref{eq:igm}).

    Therefore, we have $Q(\Theta) \in \mathcal R(\bm a^*)$ for all $\Theta \in O$.
    And consequently $\mathcal L(\Theta)  = L(Q(\Theta))\ge L(f_{\bm a^*}) = \mathcal L(\Theta_{\bm a^*})$, which means we find a local minimum $\Theta_{\bm a^*}$ for $\mathcal L$.

    In this way, let $S=\{\Theta_{\bm a^*}:\bm a^*\not\in \argmax T(\cdot)\}$.
    When initialized in $S$, gradient descent will get stuck since the gradient is zero at all local minima, resulting in a suboptimal policy.
\end{proof}


Comparing \cref{thm:vd-fails} and \cref{thm:igm-complete-vd-fails}, we shall see that the statement for VD with a complete class and that for VD with an IGM-complete class only differ in the order of the logical predicates.
That is, for any VD with a complete class, we can construct tasks to make it fail, while we can construct tasks make all VD fail uniformly when an IGM-complete class is equipped.

This difference might not be essential, in that if we add some regularity constraint for VD with a complete class (e.g., some compactness assumptions), we can still find tasks uniformly make VD fail (\cref{thm:complete-compact-vd-fails}).

Based on the principle of conducting intuitions, we are not going deep.


\begin{theorem}
    \label{thm:complete-compact-vd-fails}
    There are tasks such that, when the parameter $\Theta$ is initialized in certain set $S$, any VD with an \textbf{complete} function class as well as a $Q$ such that its inversion keeps boundedness (i.e., for all bounded $X$, $Q^{-1}(X)$ is also bounded) converges to a suboptimal policy.
\end{theorem}

\begin{proof}
    We continue to use notations defined in the proof of \cref{thm:igm-complete-vd-fails}.
    
    The argument in the proof of \cref{thm:igm-complete-vd-fails} fails here since the factorization $f_{\bm a^*} = Q(\Theta_{\bm a^*})$ cannot be guaranteed to have $\{\bm a^*(i)\} = \argmax Q_i(\cdot)$ for $i=1,\cdots,n$ when $f_{\bm a^*}$ has multiple greedy actions (i.e., $|\argmax f_{\bm a^*}(\cdot)| > 1$).

    To solve the issue, we construct a special matrix game, where
    
    $$T(\bm a) = \left\{\begin{array}{ll}
        i, & a_1=a_2=\cdots=a_n=i \\
        0, & \text{otherwise}
    \end{array}\right.$$

    Consider all $\bm a^*$ in form of $(k,\cdots,k)$ for $k=1,\cdots, |\mathcal A|-1$.

    Let $f_m(\bm a) = \left\{\begin{array}{ll}
        f_{\bm a^*}(\bm a)-\frac 1 m, & \bm a^*\neq \bm a\in\argmax f_{\bm a^*}(\cdot) \\ 
        f_{\bm a^*}(\bm a),&\text{otherwise}
    \end{array}\right.$ for $m\in \mathbb N$.

    Then $\argmax f_m(\cdot) = \{\bm  a^*\}$ for all $m\in \mathbb N$.

    By the completeness of Q-function class, we are able to find some $\Theta_m$, such that $f_m =  Q(\Theta_m) $.
    By the uniqueness of greedy policy of $f_m$, we have for all $a_i\in \mathcal A$ other than $k$: $Q_i(k;\Theta_m)>Q_i(a_i;\Theta_m)$ for $i=1,\cdots,n$.

    Consider the set $R = Q^{-1}\left(\overline B_{1}(f_{\bm a^*})\right)$.\footnote{Here $\overline B_1(f_{\bm a^*})$ is the closed Ball centered at $f_{\bm a^*}$ with radius $1$ (w.r.t. the $L^\infty$ distance $\|\cdot\|_\infty$).} $R$ is closed since $Q$ is continuous. And $R$ is also bounded, since $Q^{-1}$ keeps boundedness. 

    In this way, $R$ is compact. By Bolzano Weierstrass Theorem, we can find a convergent subsequence $\{\Theta_{m_j}\}_{j=1}^{\infty}$.
    Take the limit, we have $\Theta_{m_j}\rightarrow \Theta_{\bm a^*}\in R$.
    By the continuity, we have $Q(\Theta_{\bm a^*}) = \lim\limits_{m\rightarrow \infty} f_m = f_{\bm a^*}$, and $\forall a_i\in \mathcal A:Q_i(k;\Theta_{\bm a^*})\ge Q_i(a_i;\Theta_{\bm a^*})$ for $i=1,\cdots,n$, which means $k\in \argmax Q_i(\Theta_{\bm a^*})$

    According to the argument of \cref{thm:igm-complete-vd-fails}, the only thing remains for us is to prove $\{k\} = \argmax Q_i(\Theta_{\bm a^*})$ for all $i\in\{1,\cdots,n\}$.

    To prove it, we need the following lemma.

    \begin{lemma}
        \label{app-lemma:fa-diag-gt0-nondiag-eq0}
        For the $n$-dimensional matrix $T$ defined above, any $k\in\{1,\cdots,|\mathcal A|\}$ and $\bm a^* = (k,\cdots,k)$, we have

        \begin{enumerate}
            \item $f_{\bm a ^*}(\bm a) > 0$ for all $\bm a$ on the diagonal (i.e., $\bm a = (t,\cdots, t)$ for some $t \in \{1,\cdots,|\mathcal A|\}$).
            \item $f_{\bm a^* }(\bm b) = 0$ for all $\bm b$ not on the diagonal (i.e., $\bm b \neq (t,\cdots, t)$ for any $t\in \{1,\cdots,|\mathcal A|\}$).
        \end{enumerate}
    \end{lemma}

    \begin{proof}
        First, we can claim $f_{\bm a ^*} (\bm a^*) \ge  1$.
        Since if $f_{\bm a^*}(\bm a^*) < 1$, we can let $g(\bm u) = \left\{\begin{array}{ll}
            1 ,& \bm u = \bm a^* \\
            f_{\bm a^*} (\bm u) ,& \text{otherwise} 
        \end{array}\right.$, and in this way we have

        \begin{enumerate}
            \item $g\in \mathcal R(\bm a^*)$ in that $g(\bm a^*) = 1\ge f_{\bm a ^*}(\bm a^*) \ge f_{\bm a^*}(\bm u) = g(\bm u)$ for all $\bm u\neq \bm a^*$.

            \item $L(g) < L(f_{\bm a^*})$ in that
            \begin{enumerate}
                \item  $L(g) - L(f_{\bm a^*}) = \frac 1 2 \mathcal D(\bm a^*) [(1-T(\bm a^*))^2 - (f_{\bm a^*}(\bm a^*)-T(\bm a ^*))^2]$.
                \item $f_{\bm a^*}(\bm a^*) < 1\le T(\bm a^*)$.
            \end{enumerate}
        \end{enumerate}

        which contradicts the definition of $f_{\bm a^*}$.

        Secondly, we prove the original lemma again by contradiction similarly.

        Suppose we have $f_{\bm a^*}(\bm a) \le 0$ for some $\bm a$ on the diagonal, we can let $g(\bm u) = \left\{\begin{array}{ll}
            1 ,& \bm u = \bm a \\
            f_{\bm a^*} (\bm u) ,& \text{otherwise} 
        \end{array}\right.$. This also results in $g\in \mathcal R(\bm a^*)$ and $L(g)< L(f_{\bm a^*})$, which contradicts the definition of $f_{\bm a^*}$.

        Suppose we have $f_{\bm a ^* }(\bm b) > 0$ for some $\bm b$ not on the diagonal, we can let $g(\bm u) = \left\{\begin{array}{ll}
            0 ,& \bm u = \bm b \\
            f_{\bm a^*} (\bm u) ,& \text{otherwise} 
        \end{array}\right.$. This also results in $g\in \mathcal R(\bm a^*)$ and $L(g)< L(f_{\bm a^*})$, which contradicts the definition of $f_{\bm a^*}$.

    \end{proof}

    Then we prove it by contradiction.
    Suppose that there are some $i$, and $a_i\neq k$, such that $Q_i(k;\Theta_{\bm a^*}) = Q_i(a_i;\Theta_{\bm a^*})$.
    Let $\bm b = (\underbrace{k,\cdots,k}_{i-1 \text{ copies}},a_i,k,\cdots,k)$ differing with $\bm a^*$ only at the $i$-th action.
    Then we have $  Q(\bm b;\Theta_{\bm a^*}) = f_{\bm a^*}(\bm b) = 0$ by \cref{app-lemma:fa-diag-gt0-nondiag-eq0}, since $\bm b$ does not belong to the diagonal.
    On the other hand, $ Q(\bm b;\Theta_{\bm a^*})= \max Q(\cdot;\Theta_{\bm a^*}) = \max f_{\bm a^*}(\cdot) > 0$ by the IGM assumption, which leads to a contradiction.

    Collecting all $\Theta_{\bm a^*}$ into a set $S$ completes our proof.
\end{proof}


Now we are able to prove \cref{thm:vd-fails}.

\begin{reptheorem}
    {thm:vd-fails}
    For any VD (\cref{eq:vd-loss,eq:igm}) with a \textbf{complete} function class, there are tasks such that, when the parameter $\Theta$ is initialized in certain set $S$, it converges to a suboptimal policy.
\end{reptheorem}

\begin{proof}
    We continue to use notations in the proofs of \cref{thm:igm-complete-vd-fails,thm:complete-compact-vd-fails}.
    And we further assume the distribution $\mathcal D$ in the loss function to be uniform WLOG, since treating an arbitrary $\mathcal D$ with full support only brings complexity in form for our analysis.

    Note that in the proof of \cref{thm:complete-compact-vd-fails,thm:igm-complete-vd-fails}, it is essential for us to find $\bm a^*$ and $\Theta_{\bm a^*}$ subject to

    \begin{align}
    \label{eq:high-level-local-minima}
    \text{High level local minima:}\quad&Q(\Theta_{\bm a^*}) = f_{\bm a^*}\\
    \label{eq:decomposition-unique-and-consistent}
     \text{Greedy action consistency:}\quad&\{ \bm a^*\}=\bigtimes_{i=1}^n \argmax Q_i(\Theta_{\bm a^*})
    \end{align}

    In the proof of \cref{thm:igm-complete-vd-fails}, we fix $\bm a^*$ at first and then make it by the assumption of IGM-completeness.
    While in the proof of \cref{thm:complete-compact-vd-fails}, we again fix $\bm a^*$ at first and then make it by the B-W theorem based on the compactness of $R$.
    Nevertheless, we have neither of them here.

    However, the proof of \cref{thm:complete-compact-vd-fails} gives us an important implication:

    \begin{fact}
        \label{app-fact:diagonal-decomp-unique-argmax}
        For any \textbf{positive diagonal tensor} $f\in \mathcal R^{|\mathcal A|^n}$, i.e.,
        \begin{enumerate}
            \item $f$ is positive on the diagonal, i.e. $f(\bm a) > 0$ for all $\bm a \in\{ (k,\cdots, k): k = 1,\cdots,|\mathcal A|\}$.
            \item $f$ is zero on the non-diagonal, i.e. $f(\bm b) = 0$ for all $\bm b\not \in\{ (k,\cdots,k):k = 1,\cdots,|\mathcal A|\}$.
        \end{enumerate}

        we have $|\argmax Q_i(\Theta)| = 1$ for $i  =1,\cdots,n$, where $\Theta$ is the parameter that decomposes $f$ (i.e., $Q(\Theta) = f$).
    \end{fact}
    
    \begin{proof}
        This is essentially the argument at the end of the proof of \cref{thm:complete-compact-vd-fails}. We omit the complete proof here for simplicity.
    \end{proof}

    \cref{app-fact:diagonal-decomp-unique-argmax} tells us that factorizing a \textbf{positive diagonal tensor} results in a unique local greedy action for all $i\in[n]$.

    Now we are going to construct some positive diagonal tensor $T$ with a proper structure.
    We first define a recursive sequence $\{h_t\}_{t\ge 1}$ as follows:

    \begin{itemize}
        \item $h_1=0$.
        \item $\forall t>1: h_t = (t-1)h_{t-1} - \sum\limits _{i<t-1} h_i - 1$.
    \end{itemize}

    Let $C\in \mathbb R^{|\mathcal A|}$ and $C(t) = h_{|\mathcal A| - t + 1} -  h_{|\mathcal A|} + 1$ for $t = 1,\cdots, |\mathcal A|$, we fill the diagonal of the $n$-dimensional tensor $T$ by $C(1),\cdots,C(|\mathcal A|)$, denoted by $T = \text{diag}_n\{C\}$.
    We have the following lemma for the structure of $f_{\bm a^*}$ when $\bm a^*$ is some joint action on the diagonal.

    \begin{lemma}
        \label{app-lemma:fa-star-structure}
        Suppose $T = \text{diag}_n\{C\}$ and $\bm a^* = (l,\cdots, l)$ for some $l\in\{1,\cdots, |\mathcal A|\}$, we have $f_{\bm a^*} = \text{diag}_n \{D\}$, where

        \begin{enumerate}
            \item $D(k) = C(k)$ for $k = 1,\cdots, l-1$.
            \item $D(k) = \frac 1 {|\mathcal A| - l + 1} \sum _{i\ge l} C(i)$ for $k = l,\cdots,|\mathcal A|$.
        \end{enumerate}

        Moreover, let $P$ be any permutation over $\{1,\cdots,|\mathcal A|\}$, then if $T = \text{diag}_n\{P\circ C\}$, we have $f_{\bm a ^*}=\text{diag}_n\{P\circ D\}$.\footnote{$P\circ A$ means the permutation $P$ acting on the vector $A$, which results in a new vector $B$ such that $B(k) = A(P(k))$.}
    \end{lemma}

    \begin{proof}
        It is easy to see that $C(k)\ge 1$ for all $k$.
        Therefore by \cref{app-lemma:fa-diag-gt0-nondiag-eq0}, $f_{\bm a^*}$ should be a diagonal tensor.

        Let $D(l) = v$, by the definition of $f_{\bm a^*}$, we have $D(k)\le v$ for all $k$.
        Suppose $C(k)<v$ for some $k\neq l$, then we must have $D(k) = C(k)$ since this would result in a zero loss on the entry; and for all $k\neq l$ such that $C(k)\ge v$, we must have $D(k) = v$ since this would result in a minimum loss under the constraint $D(k)\le v$.
        
        Let $M(v)\in\mathcal R^{|\mathcal A|}$ be a vector such that $M(v)(k)=\left\{\begin{array}{ll}
            C(k), & C(k)<v \text{ and } k\neq l \\
            v, & \text{otherwise}
        \end{array}\right.$

        Then minimizing $L$ in $\mathcal R(\bm a^*)$ is equivalent to minimizing $\phi(v) = L(M(v))$.

        Now we show that the global minimum $ v^*$ of $L$ is $\frac 1 {|\mathcal A| - l + 1} \sum _{i\ge l} C(i)$.

        First, it is obvious that $v^*\ge C(l)$.

        Secondly, note that

        \begin{align*}
            \left[\frac{1}{|\mathcal A|-l+1}\sum _{i\ge l } C(i)\right] - C(l+1)& =  \left[\frac{1}{|\mathcal A|-l+1} \sum_{i\ge l }(h_{|\mathcal A| - i+1}-h_{|\mathcal A|} +1) \right]- (h_{|\mathcal A| - l}-h_{|\mathcal A|} +1 )\\
            & =  \left[\frac{1}{|\mathcal A|-l+1} \sum_{i\le |\mathcal A|-l+1 }h_{i}\right]- h_{|\mathcal A| - l}\\
            & = \frac{1}{|\mathcal A|-l+1} \left[(|\mathcal A|-l)h_{|\mathcal A|-l}-\sum_{i\le |\mathcal A|-l-1} h_i-1  + \sum_{i\le |\mathcal A|-l }h_{i}\right]- h_{|\mathcal A| - l}\\
            & = \frac{1}{|\mathcal A|-l+1} \left[(|\mathcal A|-l + 1)h_{|\mathcal A|-l}-1\right]- h_{|\mathcal A| - l}\\
            & = -\frac{1}{|\mathcal A| - l + 1}
        \end{align*}

        That is, $\mathbf{Mean}\{C(l),\cdots,C(|\mathcal A|)\}< C(l+1)$.
        As a consequence, $\mathbf{Mean}\{C(l),C(l+2),\cdots, C(|\mathcal A|)\}<C(l+1)$.
        By induction, it is easy to show that $\mathbf{Mean}\{C(l),C(l'),\cdots,C(|\mathcal A|)\}<C(l+1)$ for all $l'>l$.

        Third, note that for all $v\ge C(l)$,

        $$\phi(v) = \frac 1 { |\mathcal A|^n}\left( (C(l)-v)^2 + \sum_{i\ge l'}(C(i)- v)^2\right)\quad\text{where }l'=\min\{j:C(j)\ge v\}\footnote{The term $\frac 1 {|\mathcal A|^n}$ comes from the uniform distribution $\mathcal D$.}$$

        Also, recall that for any $m$ and $x_1,\cdots, x_m$, the function $\xi(x)\stackrel{\Delta} = \sum_{k=1}^m(x_k-x)^2$ monotonically decreases on $(-\infty, \overline x)$ and monotonically increases on $(\overline x, +\infty)$, where $\overline x = \mathbf{Mean}\{x_1,\cdots,x_m\}$.\footnote{It can be easily proved by taking the derivative.}

        Thus, we have

        \begin{enumerate}
            \item $\phi$ is continuous.
            \item On $[C(l),\mathbf{Mean}\{C(l),\cdots,C(|\mathcal A|)\})$, $\phi(v) =\frac 1 {|\mathcal A|^n}\sum_{i\ge l}(C(i)- v)^2$ monotonically decreases.
            \item On $(\mathbf{Mean}\{C(l),\cdots,C(|\mathcal A|)\}, C(l+1)]$, $\phi(v) =\frac 1 {|\mathcal A|^n}\sum_{i\ge l}(C(i)- v)^2$ monotonically increases.
            \item On $(C(l+1), C(l+2)]$, $\phi(v) =\frac 1 {|\mathcal A|^n}\left((C(l)-v)^2 + \sum_{i\ge l+2}(C(i)- v)^2\right)$ monotonically increases.
            \item Based on 1, 3, 4 and by induction, on $(\mathbf{Mean}\{C(l),\cdots,C(|\mathcal A|)\}, + \infty)$, $\phi(v)$ monotonically increases.
        \end{enumerate}

        which means $v^* = \mathbf{Mean}\{C(l),\cdots,C(|\mathcal A|)\}$.

        This completes the proof of the first part of the lemma.
        For the second part when a permutation $P$ is involved, the discussion is essentially the same, we omit it here.

    \end{proof}

    \cref{app-lemma:fa-star-structure} shows that when we choose $\bm a^* = (l,\cdots,l)$, the resulting diagonal of $f_{\bm a^*}$ contains two parts:

    \begin{enumerate}
        \item The top $l-1$ smallest entries are ``revealed" as they exactly are.
        \item Other entries are ``masked" by the same constant $v = \frac 1 {|\mathcal A|-l+1}\sum_{i\ge l}C(i)$.
    \end{enumerate}

    This motivates us to construct the diagonal of $T$ recursively.

    \begin{algorithm}[htb]
      \caption{Construction of $T$ and Local Minima}
      \label{algdef:construct-T-local-minima}
      
    \begin{algorithmic}[1]
      \STATE {\bfseries Output:} $T\in \mathbb R^{|\mathcal A|^n}$ and a set $S$ of parameters.
      \STATE Initialize $E \in \mathbb R^{|\mathcal A|}$ with all zeros.\quad\COMMENT{The diagonal of $T$ to be constructed.}
    
      \FOR{$l = 1$ \textbf{to} $|\mathcal A|$}
        \STATE $D\gets E$.\quad\COMMENT{Reveal $l-1$ known entries of the diagonal.}
        \STATE Fill in all zero entries of $D$ by $v = \frac{1}{|\mathcal A| - l + 1}\sum_{i\ge l} C(i)$.\quad\COMMENT{Mask all unknown entries.}
        \STATE $f\gets \text{diag}_n\{D\}$.
        \STATE Calculate $\Theta_l$ such that $Q(\Theta_l) = f$.\quad\COMMENT{Decomposition, by the completeness assumption.}
        \STATE $S = S\cup \{\Theta_l\}$.\quad\COMMENT{This would be a local minimum at last.}
        \STATE $j \gets \argmax Q_1(\Theta_l)$.\quad\COMMENT{By \cref{app-fact:diagonal-decomp-unique-argmax}, $\argmax Q_1(\Theta_l)=\cdots = \argmax Q_{n}(\Theta_l)$, $|\argmax Q_1(\Theta_l)| = 1$. }
        \STATE $E(j) \gets C(l)$.\quad\COMMENT{Determine the index of $C(l)$.}
      \ENDFOR
      \STATE $T\gets  \text{diag}_n\{E\}$.
      \STATE {\bfseries Return} $T,S$.
    \end{algorithmic}
    \end{algorithm}

    It is easy to check that \cref{eq:high-level-local-minima,eq:decomposition-unique-and-consistent} are satisfied by the $T$ and $S$ output by \cref{algdef:construct-T-local-minima}.

    This completes our proof.
\end{proof}

\subsection{Some Discussion on QPLEX}

QPLEX\citep{Wang2021QPLEXDD} is a value-decomposition algorithm with an IGM-complete (\cref{def:complete-igm-complete}) function class.
According to \cref{thm:vd-fails}, its loss function $\mathcal L$ has numerous local minima.

In the empirical design of the algorithm, we notice that QPLEX has used some engineering tricks like "stop gradient" to modify the gradient of non-optimal points helping the algorithm to jump out of local optima.
But these tricks lack of theoretical guarantee, we can still construct cases where QPLEX is not able to reach the global optimum, such as the following Matrix Game (Table\ref{matgame2}).

\begin{table}[htb]
    \centering
    \begin{tabular}{|l|l|}
    \hline
    -20   & 10  \\ \hline
    10 & 9   \\ \hline

    \end{tabular}
    \caption{Matrix Game 2, $m=2$}
    \label{matgame2}
\end{table}

This Matrix Game has two global optima $(0,1)$ and $(1,0)$, and one suboptimal solution $(1,1)$ with high reward.
QPLEX will likely to initialize to the suboptimal solution $(1,1)$, and after that, it gets confused since the manually modified gradient does not tell it the right direction.
The learned joint $Q$ vibrates around the following matrix:

$$\left(\begin{array}{cc}-20&29/3-\epsilon\\29/3-\epsilon &29/3\end{array}\right)$$

which can be proved to be a local optimum of the loss function $\mathcal L$ according to the proof of \cref{thm:vd-fails}.

\section{Omitted Definitions and Proofs in Section  \ref{sec:transformation}}
\label{app:transformation}

\subsection{Omitted proof of Theorem \ref{thm:convert-policy-keeps-value}}

Before we provide the proof of \cref{thm:convert-policy-keeps-value}, we first introduce the policy conversion between a policy $\pi$ on $\Gamma(\mathcal M)$ and a coordination policy $\pi^c$ on $\mathcal M$. $\mathcal M$ and $\Gamma(\mathcal M)$ refer to an MMDP and its sequential transformation (\cref{algdef:transformation}).

We define the coordination policy at first.

\begin{definition}[coordination policy]
\label{def:coordination-policy}
A coordination policy $ \pi^c$ is $n$ individual policies for sequential decision making. That is $\pi^c = (\pi_1,\cdots,\pi_n)$, where $\pi_i(a_i|s,a_{<i})$ is a policy of agent $i$ conditioned on the state $s$ and ``previously inferred" actions $a_{<i} = \langle a_1,\dots, a_{i-1}\rangle$.
\end{definition}


Then there is a natural correspondence between the policy $\pi$ on $\Gamma(\mathcal M)$ and the corresponding coordination policy $\pi^c$ on $\mathcal M$.
We present the bi-directional conversion as follows.

\paragraph{from $\pi$ to $\pi^c$} Suppose we have a policy $\pi(a_{k+1}|(s,a_1,\cdots,a_k))$ on $\Gamma(\mathcal M)$.
Denote $\pi_k(a_{k}|s,a_1,\cdots,a_{k-1}) = \pi(a_{k}|(s,a_1,\cdots,a_{k-1}))$.
Then $\pi^c = (\pi_1,\cdots,\pi_n)$ is the corresponding coordination policy on $\mathcal M$.

\paragraph{from $\pi^c$ to $\pi$} Suppose we have a coordination policy $\pi^c = (\pi_1,\cdots,\pi_n)$ on $\mathcal M$.
We can directly have the policy $\pi$ on $\Gamma(\mathcal M)$, where $\pi(a_{k}|(s,a_1,\cdots,a_{k-1})) = \pi_k(a_{k}|s,a_1,\cdots,a_{k-1})$.

Now, we rephrase \cref{thm:convert-policy-keeps-value} and give a proof.

\begin{reptheorem}{thm:convert-policy-keeps-value}
For any policy $\pi$ on $\Gamma\left(\mathcal M\right)$ and its corresponding coordination policy $\pi^c$ on $\mathcal M$, we have $\mathcal J_{\mathcal M}(\pi^c) = \gamma^{(1-n)/n}\mathcal  J_{\Gamma\left(\mathcal M\right)}(\pi)$.
\end{reptheorem}

\begin{proof}

We use the notations defined in \cref{algdef:transformation}.

For some policy $\eta$, the notation $r(s,\eta)$ refers to $\mathbb E_{a\sim \eta(\cdot|s)}[r(s,a)]$.
And the notation $P(s,\eta)$ refers to the distribution of $s'$, where $a\sim \eta(\cdot|s), s'\sim P(\cdot|s,a)$.

\begin{align*}
    \mathcal J_{\mathcal M}(\pi^c)& = 
        \mathbb E\left[
            \sum_{t=0}^\infty\gamma^tr(s_t,\pi^c)
            \middle|
            s_{t+1}\sim P(s_t,\pi^c)
        \right]\\
    & = \mathbb E\left[
        \sum_{t=0}^\infty   {\gamma'}^{nt}
            r'\left(\left(s_t,a_{<n}^{(t)}\right) ,\pi \right)
        \middle| 
        s_{t+1}\sim P(s_t,\pi^c),
        a_l^{(t)}\sim \pi\left(\cdot \middle| \left(s_t,a_{<l}^{(t)}\right)\right)
    \right]\\
    &=\mathbb E\left[
        \sum_{t=0}^\infty \sum_{k=0}^{n-1}{\gamma'}^{nt+k-n+1} 
            r'\left(\left(s_t,a_{<k}^{(t)}\right), \pi \right)
        \middle|
        s_{t+1}\sim P(s_t,\pi^c),
        a_l^{(t)}\sim 
            \pi\left(\cdot\middle|\left(s_t,a_{<l}^{(t)}\right)\right)\right]\\
    &=\mathbb E\left[
        {\gamma'}^{1-n}\sum_{t'=0}^\infty
            {\gamma'}^{t'} r'(s_{t'},\pi)
        \middle|
            s_{t'+1}\sim  P'(s_{t'},\pi)
        \right]\\
    &=\mathcal \gamma^{\frac{1-n}{n}}J_{\Gamma\mathcal (\mathcal M)}(\pi)
\end{align*}

\end{proof}

\subsection{Omitted Proof of Theorem \ref{thm:TAD-keeps-optimality}}

Before we prove \cref{thm:TAD-keeps-optimality}, we first define our distillation process on MMDP in \cref{algdef:distillation}.

\begin{algorithm}[htb]
  \caption{The Distillation Stage in the TAD framework}
  \label{algdef:distillation}
  
\begin{algorithmic}[1]
  \STATE {\bfseries Input:} An MMDP $\mathcal M$ of $n$ agents, a coordination policy $\pi^c=(\pi_1,\cdots, \pi_n)$ on $\mathcal M$.
  \STATE {\bfseries Output:} Decentralized policies $\eta_1,\cdots,\eta_n$.

  \STATE\COMMENT{Calculate the greedy policy for each agent.}
  \FOR{$k = 1$ \textbf{to} $n$}
     \FORALL{$s\in \mathcal S, a_1,\cdots,a_{k-1}\in \mathcal A$}
     \STATE $a^*\gets \argmax_{a_k} \pi_k(a_k|s,a_1,\cdots,a_{k-1}) $\quad\COMMENT{If there are multiple, select arbitrarily.}
     \STATE $\mu_k(s,a_1,\cdots,a_{k-1})=a^*$.\quad\COMMENT{Convert $\pi^c$ to a deterministic coordination policy.}
     \ENDFOR
  \ENDFOR

  \FORALL{$s\in \mathcal S$}
    \FOR{$k=1$ \textbf{to} $n$}
        \STATE $\eta_k(s) = \mu_k(s,\eta_1(s),\cdots,\eta_{k-1}(s))$.\quad\COMMENT{The deterministic coordination policy is natually decentralized.}
    \ENDFOR
  \ENDFOR
  
  \STATE {\bfseries Return} $\eta_1,\cdots,\eta_n$.
\end{algorithmic}
\end{algorithm}

Now we present the proof of \cref{thm:TAD-keeps-optimality}.

\begin{reptheorem}
    {thm:TAD-keeps-optimality}

    For any SARL algorithms $A$ with optimality guarantee on MDP, TAD-$A$, the MARL algorithm using TAD framework and based on $A$, will have optimality guarantee on any MMDP.
\end{reptheorem}

\begin{proof}
    For arbitrary MMDP $\mathcal M$, since $A$ has an optimality guarantee on MDP, the policy on $\Gamma(\mathcal M)$ learned by $A$ maximizes the expected return on $\Gamma(\mathcal M)$. Then by \cref{thm:convert-policy-keeps-value}, the coordination policy $\pi^c$ obtained in the transformation stage (see \cref{alg:TAD-framework}) maximizes the expected return on $\mathcal M$.

    In the distillation stage (\cref{algdef:distillation}), $\pi^c$ is first converted into a deterministic coordination policy $\mu^c = (\mu_1,\cdots,\mu_k)$.
    Since $\pi^c$ is optimal on $\mathcal M$, $\mu^c$ should also be optimal.

    Finally, line 8-10 in \cref{algdef:distillation} describe the way we decentralize a deterministic coordination policy $\mu^c$.
    It is obvious that $\mu^c = \bigtimes_{i=1}^n\eta_i$.
    Therefore the decentralized policies $\eta_1,\cdots,\eta_n$ are also optimal.
\end{proof}

\begin{theorem}
    \label{thm:TAD-PPO-optimality}
    Suppose that assumptions
4.1, 4.3, and 4.4 in \citet{ppooptimal} hold. TAD-PPO has optimality guarantee on finite multi-agent MDPs.
\end{theorem}

\begin{proof}
     By Theorem 4.10 in \citet{ppooptimal}, under the assumptions
4.1, 4.3, and 4.4 in \citet{ppooptimal}, A properly implemented PPO (even with a neural network), achieves global optimality \citep{ppooptimal}.  TAD-PPO converges to the global optimum by our \cref{thm:TAD-PPO-optimality}.
\end{proof}


\subsection{The Complexity of the Transformed Model}
\label{subsec:complexity}

By the TAD framework, we are able to convert any MMDP to an MDP and run SARL algorithms on the MDP to solve the MMDP. One natural question is, will such framework bring additional hardness of the task?

We will exclude the influence of the distillation stage and consider the sequential transformation only.
That is, we consider the sequential transformation framework (the TAD framework without distillation).

First of all, the sequential transformation framework is reversible in that for any multi-agent algorithm $A$, and any MDP $\widetilde{\mathcal M}$, we can always compress $n$ steps on $\widetilde{\mathcal M}$ into one step to obtain $\Gamma^{-1}(\widetilde{\mathcal M})$ and then use the corresponding multi-agent algorithm $A$ to solve it.

Then, we introduce the concept of minimax sample complexity.
Roughly speaking, the minimax sample complexity is the sample complexity of the ``best" algorithm over the ``hardest" task.
Since the sequential transformation framework is reversible with only negligible additional cost in time and space, we are able to claim the sequential transformation framework does not increase the minimax sample complexity\footnote{One should keep in mind that every $n$ samples on MDP correspond to exactly one sample on MMDP. Particularly, the number of bits we need to record every $n$ samples on MDP are exactly what we need to record one sample on MMDP.}.

Nevertheless, for a concrete algorithm $A$ (e.g. Q-learning), the sample complexity is not necessary to be the same after such transformation.

We take Q-learning as an example and first investigate the size of state-action space before and after the transformation.
It easy to see that the size of the state-action space of $\mathcal M$ is $|\mathcal S||\mathcal A|^n$, and that of its $\Gamma\left(\mathcal M\right)$ is $|\mathcal A|\sum_{i=0}^{n-1}|\mathcal S||\mathcal A|^i = |\mathcal S||\mathcal A|\frac{1-|\mathcal A|^n}{1-|\mathcal A|} \le 2|\mathcal S||\mathcal A|^N$.
This implies that the sequential transform does not increase the complexity in the state-action space.

However, if we take a closer look here of the sample complexity, we will find that the exact sample complexity bound of Q-learning is $O^*\left(\frac{|\mathcal S||\mathcal A|}{(1-\gamma)^4\epsilon^2}\right)$ (\citep{li2021qlearning}), which depends on not only the size of state-action space, but also on the magnitude of $\frac 1 {1-\gamma}$.
This implies that the sample complexity may increase for certain algorithms since $\Gamma\left(\mathcal M\right)$ has a longer horizon ($n$ times).

\section{Experimental Details}\label{app:experimental_details}

In this section, we provide more experimental results supplementary to those presented in Section \ref{exp}. We also discuss the details of the experimental settings of both our matrix game and the StarCraft II micromanagement (SMAC) benchmark.



\subsection{Details of Multi-task Matrix Game}
\label{matrices}

In multi-task matrix game, the return of the optimal strategy corresponding to each matrix is 10, which means the sum rewards of the global optimal strategy is 100. Two agents are initialized to one matrix uniformly at random, and the ID of a current matrix is observable to both of them. They need to cooperate to select the entry with the maximum reward for the current matrix, after that, the game ends. Each matrix contains $5\times 5=25$ entries, which means $\mathcal A = \{0,1,2,3,4\}$ for each agent.

\begin{table}[!htb]
    \caption{Multi-task Matrix Game}
    \label{tab:msmatgame}
    \begin{minipage}{.5\linewidth}
      \caption*{Matrix 1}
      \centering
        \begin{tabular}{|c|c|c|c|c|c|}

        \hline
        \diagbox{$a_2$}{$a_1$}    & \(\mathcal{A}^{(1)}\) & \(\mathcal{A}^{(2)}\) & \(\mathcal{A}^{(3)}\) & \(\mathcal{A}^{(4)}\) & \(\mathcal{A}^{(5)}\) \\ \hline
        \(\mathcal{A}^{(1)}\) & \textbf{10}           & -10                   & -10                   & -10                   & -10                   \\ \hline
        \(\mathcal{A}^{(2)}\) & -10                   & 9                     & 0                     & 0                     & 0                     \\ \hline
        \(\mathcal{A}^{(3)}\) & -10                   & 0                     & 9                     & 0                     & 0                     \\ \hline
        \(\mathcal{A}^{(4)}\) & -10                   & 0                     & 0                     & 9                     & 0                     \\ \hline
        \(\mathcal{A}^{(5)}\) & -10                   & 0                     & 0                     & 0                     & 9                     \\ \hline
        \end{tabular}
    \end{minipage}%
    \begin{minipage}{.5\linewidth}
      \centering
        \caption*{Matrix 2}
        \begin{tabular}{|c|c|c|c|c|c|}
        \hline
        \diagbox{$a_2$}{$a_1$}                   & \(\mathcal{A}^{(1)}\) & \(\mathcal{A}^{(2)}\) & \(\mathcal{A}^{(3)}\) & \(\mathcal{A}^{(4)}\) & \(\mathcal{A}^{(5)}\) \\ \hline
        \(\mathcal{A}^{(1)}\) & \textbf{10}           & -10                   & \textbf{10}                    & -10                   & \textbf{10}           \\ \hline
        \(\mathcal{A}^{(2)}\) & -10                   & \textbf{10}           & -10                   & \textbf{10}           & -10                   \\ \hline
        \(\mathcal{A}^{(3)}\) & \textbf{10}           & -10                   & \textbf{10}           & -10                   & \textbf{10}           \\ \hline
        \(\mathcal{A}^{(4)}\) & -10                   & \textbf{10}           & -10                   & \textbf{10}           & -10                   \\ \hline
        \(\mathcal{A}^{(5)}\) & \textbf{10}           & -10                   & \textbf{10}           & -10                   & \textbf{10}           \\ \hline
        \end{tabular}
    \end{minipage} 
    \begin{minipage}{.5\linewidth}
      \caption*{Matrix 3}
      \centering
        \begin{tabular}{|c|c|c|c|c|c|}
        \hline
        \diagbox{$a_2$}{$a_1$}                  & \(\mathcal{A}^{(1)}\) & \(\mathcal{A}^{(2)}\) & \(\mathcal{A}^{(3)}\) & \(\mathcal{A}^{(4)}\) & \(\mathcal{A}^{(5)}\) \\ \hline
        \(\mathcal{A}^{(1)}\) & -20                   & -20                   & -20                   & -20                   & \textbf{10}           \\ \hline
        \(\mathcal{A}^{(2)}\) & -20                   & -20                   & -20                   & \textbf{10}           & 9                     \\ \hline
        \(\mathcal{A}^{(3)}\) & -20                   & -20                   & \textbf{10}           & 9                     & 9                     \\ \hline
        \(\mathcal{A}^{(4)}\) & -20                   & \textbf{10}           & 9                     & 9                     & 9                     \\ \hline
        \(\mathcal{A}^{(5)}\) & \textbf{10}           & 9                     & 9                     & 9                     & 9                     \\ \hline
        \end{tabular}
    \end{minipage}%
    \begin{minipage}{.5\linewidth}
      \centering
        \caption*{Matrix 4}
        \begin{tabular}{|c|c|c|c|c|c|}
        \hline
        \diagbox{$a_2$}{$a_1$}                  & \(\mathcal{A}^{(1)}\) & \(\mathcal{A}^{(2)}\) & \(\mathcal{A}^{(3)}\) & \(\mathcal{A}^{(4)}\) & \(\mathcal{A}^{(5)}\) \\ \hline
        \(\mathcal{A}^{(1)}\) & -20                   & -20                   & -20                   & -20                   & \textbf{10}           \\ \hline
        \(\mathcal{A}^{(2)}\) & -20                   & -20                   & -20                   & \textbf{10}           & 9                     \\ \hline
        \(\mathcal{A}^{(3)}\) & -20                   & -20                   & \textbf{10}           & 9                     & 8                     \\ \hline
        \(\mathcal{A}^{(4)}\) & -20                   & \textbf{10}           & 9                     & 8                     & 7                     \\ \hline
        \(\mathcal{A}^{(5)}\) & \textbf{10}           & 9                     & 8                     & 7                     & 6                     \\ \hline
        \end{tabular}
    \end{minipage} 
        \begin{minipage}{.5\linewidth}
      \caption*{Matrix 5}
      \centering
        \begin{tabular}{|c|c|c|c|c|c|}
        \hline
        \diagbox{$a_2$}{$a_1$}                  & \(\mathcal{A}^{(1)}\) & \(\mathcal{A}^{(2)}\) & \(\mathcal{A}^{(3)}\) & \(\mathcal{A}^{(4)}\) & \(\mathcal{A}^{(5)}\) \\ \hline
        \(\mathcal{A}^{(1)}\) & -20                   & -15                   & -10                   & -5                    & 6                     \\ \hline
        \(\mathcal{A}^{(2)}\) & -20                   & -15                   & -10                   & 7                     & 5                     \\ \hline
        \(\mathcal{A}^{(3)}\) & -20                   & -15                   & 8                     & 6                     & 4                     \\ \hline
        \(\mathcal{A}^{(4)}\) & -20                   & 9                     & 7                     & 5                     & 3                     \\ \hline
        \(\mathcal{A}^{(5)}\) & \textbf{10}           & 8                     & 6                     & 4                     & 2                     \\ \hline
        \end{tabular}
    \end{minipage}%
    \begin{minipage}{.5\linewidth}
      \centering
        \caption*{Matrix 6}
        \begin{tabular}{|c|c|c|c|c|c|}
        \hline
        \diagbox{$a_2$}{$a_1$}                  & \(\mathcal{A}^{(1)}\) & \(\mathcal{A}^{(2)}\) & \(\mathcal{A}^{(3)}\) & \(\mathcal{A}^{(4)}\) & \(\mathcal{A}^{(5)}\) \\ \hline
        \(\mathcal{A}^{(1)}\) &0.8       &-16.0   &-5.0      &-10.9   &-3.7                      \\ \hline
        \(\mathcal{A}^{(2)}\) &-9.2      &-4.2      &7.3       &9.6       &-3.0                      \\ \hline
        \(\mathcal{A}^{(3)}\) & -20.0      &-18.1   &0.2       &-4.3      &9.0                       \\ \hline
        \(\mathcal{A}^{(4)}\) & -14.9      &-2.0      &-17.7   &-17.6   &-0.8                      \\ \hline
        \(\mathcal{A}^{(5)}\) & 3.8      &\textbf{10}      &7.5       &9.2       &-10.7                     \\ \hline
        \end{tabular}
    \end{minipage} 
        \begin{minipage}{.5\linewidth}
      \caption*{Matrix 7}
      \centering
        \begin{tabular}{|c|c|c|c|c|c|}
        \hline
        \diagbox{$a_2$}{$a_1$}                  & \(\mathcal{A}^{(1)}\) & \(\mathcal{A}^{(2)}\) & \(\mathcal{A}^{(3)}\) & \(\mathcal{A}^{(4)}\) & \(\mathcal{A}^{(5)}\) \\ \hline
        \(\mathcal{A}^{(1)}\) & -14.4   & -15.8   & 1.5   & -5.4  & \textbf{10}                      \\ \hline
        \(\mathcal{A}^{(2)}\) & -13.2   & 5.8     & -8.7  & -2.2  & -18.2                      \\ \hline
        \(\mathcal{A}^{(3)}\) & -5.9    & -19.0   & -0.7  & -2.0  & -19.5                      \\ \hline
        \(\mathcal{A}^{(4)}\) & 0.8     & 4.7     & -14.8 & 2.5   & -4.1                     \\ \hline
        \(\mathcal{A}^{(5)}\) & -11.3   &-8.2     &-20.0  &-17.3  &-17.6                     \\ \hline
        \end{tabular}
    \end{minipage}%
    \begin{minipage}{.5\linewidth}
      \centering
        \caption*{Matrix 8}
        \begin{tabular}{|c|c|c|c|c|c|}
        \hline
        \diagbox{$a_2$}{$a_1$}                  & \(\mathcal{A}^{(1)}\) & \(\mathcal{A}^{(2)}\) & \(\mathcal{A}^{(3)}\) & \(\mathcal{A}^{(4)}\) & \(\mathcal{A}^{(5)}\) \\ \hline
        \(\mathcal{A}^{(1)}\) &-1.4   & -19.2   & 7.2   & -5.5  & 7.4                       \\ \hline
        \(\mathcal{A}^{(2)}\) &-18.5  & -20.0   & -14.4   & -17.6   & -5.1                      \\ \hline
        \(\mathcal{A}^{(3)}\) & 3.6   &5.5   & \textbf{10}  &-13.3   &-4.9                       \\ \hline
        \(\mathcal{A}^{(4)}\) & 9.8   &-12.3   &0.6   &-16.5   &-13.0                      \\ \hline
        \(\mathcal{A}^{(5)}\) & -11.8  &-20.0  &-2.4 &7.1  &-2.3                      \\ \hline
        \end{tabular}
    \end{minipage} 
        \begin{minipage}{.5\linewidth}
      \caption*{Matrix 9}
      \centering
        \begin{tabular}{|c|c|c|c|c|c|}
        \hline
        \diagbox{$a_2$}{$a_1$}                  & \(\mathcal{A}^{(1)}\) & \(\mathcal{A}^{(2)}\) & \(\mathcal{A}^{(3)}\) & \(\mathcal{A}^{(4)}\) & \(\mathcal{A}^{(5)}\) \\ \hline
        \(\mathcal{A}^{(1)}\) & -4.5    &-5.2     &-8.4   &-8.9   &5.5                       \\ \hline
        \(\mathcal{A}^{(2)}\) & -12.4   &-9.5     &8.8    &5.4    &4.4                       \\ \hline
        \(\mathcal{A}^{(3)}\) & -4.6    &1.3      &5.5    &7.3    &-6.8                       \\ \hline
        \(\mathcal{A}^{(4)}\) & 9.0     &-18.7    &-18.2  &-13.7  &-8.2                     \\ \hline
        \(\mathcal{A}^{(5)}\) & 2.2     &-9.1     &\textbf{10}    &7.1    &-20.0                      \\ \hline
        \end{tabular}
    \end{minipage}%
    \begin{minipage}{.5\linewidth}
      \centering
        \caption*{Matrix 10}
        \begin{tabular}{|c|c|c|c|c|c|}
        \hline
        \diagbox{$a_2$}{$a_1$}                  & \(\mathcal{A}^{(1)}\) & \(\mathcal{A}^{(2)}\) & \(\mathcal{A}^{(3)}\) & \(\mathcal{A}^{(4)}\) & \(\mathcal{A}^{(5)}\) \\ \hline
        \(\mathcal{A}^{(1)}\) &-8.4   &-1.8  &-20.0   &7.3   &-3.0                       \\ \hline
        \(\mathcal{A}^{(2)}\) &-8.7   &1.7   &4.8   &2.0   &-7.8                      \\ \hline
        \(\mathcal{A}^{(3)}\) & -13.3 &-3.2  &0.7   &-1.8  &-10.7                        \\ \hline
        \(\mathcal{A}^{(4)}\) & 9.8   &-12.3   &0.6   &-16.5   &-13.0                      \\ \hline
        \(\mathcal{A}^{(5)}\) & 1.8   &2.9   &-1.1  &\textbf{10}  &8.2                    \\ \hline
        \end{tabular}
    \end{minipage} 
\end{table}

All 10 payoff matrices are listed in Table \ref{tab:msmatgame}. The optimal strategies' payoff of all matrices is 10. Matrices $1-5$ are hand-crafted in order to create some hard NEs. Matrices $6-10$ are drawn uniformly at random. It's worth noting that a random $5\times 5$ matrix has $25/9\approx 2.77$ different NEs in expectation. And in our opinion, the existence of suboptimal NEs is the main reason why existing algorithms fail.

\newpage

\subsection{Benchmarking on StarCraft II Micromanagement Tasks}
\label{sec:smac_all}

\begin{figure}[t]
  \centering
    \includegraphics[width=.85\linewidth]{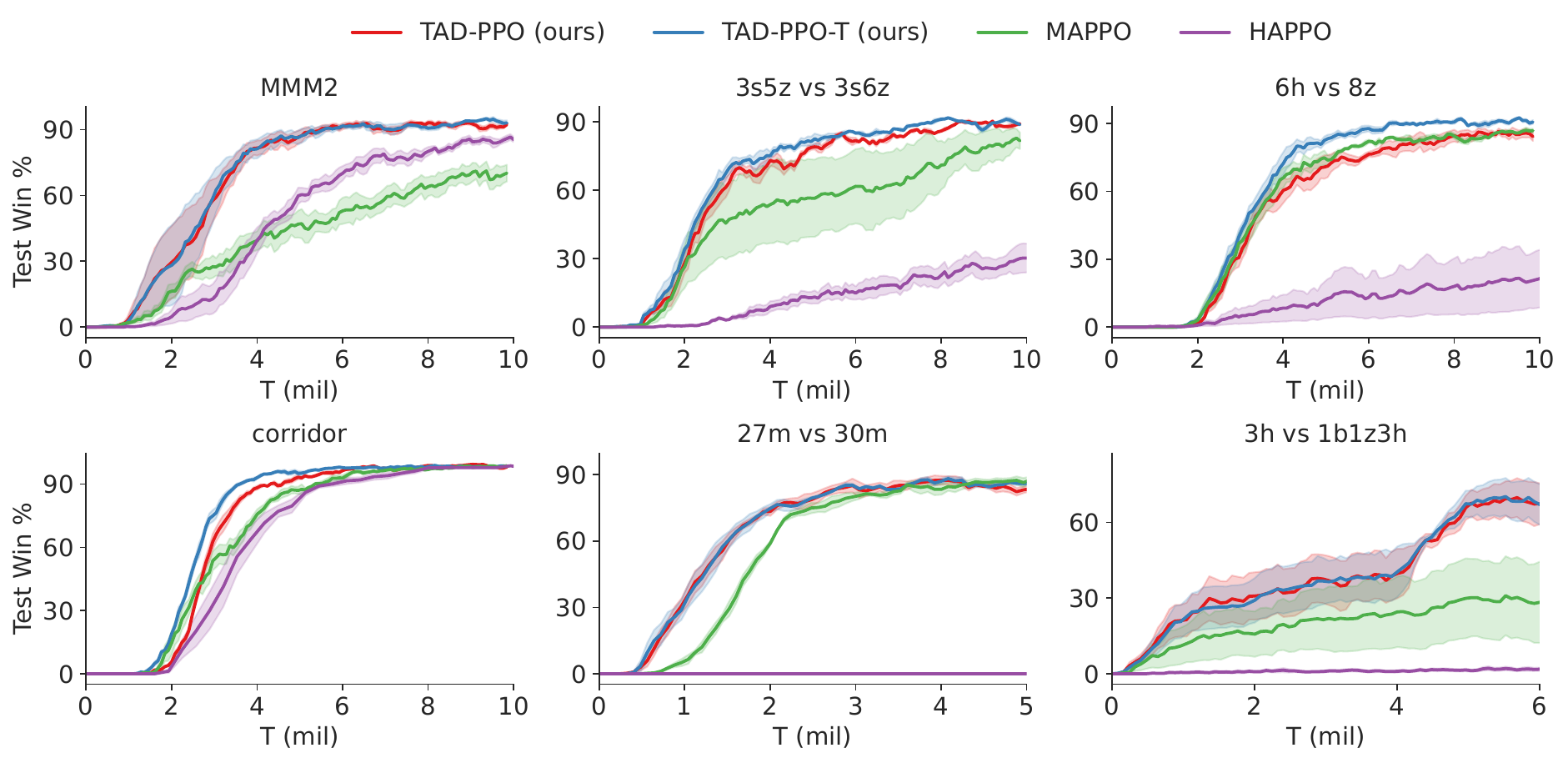}
  \caption{Comparisons between our approach and policy-based state-of-the-art MARL baselines on all \textbf{super hard} maps.}
  \label{fig:superhard}
\end{figure}

\begin{figure}[t]
  \centering
    \includegraphics[width=.6\linewidth]{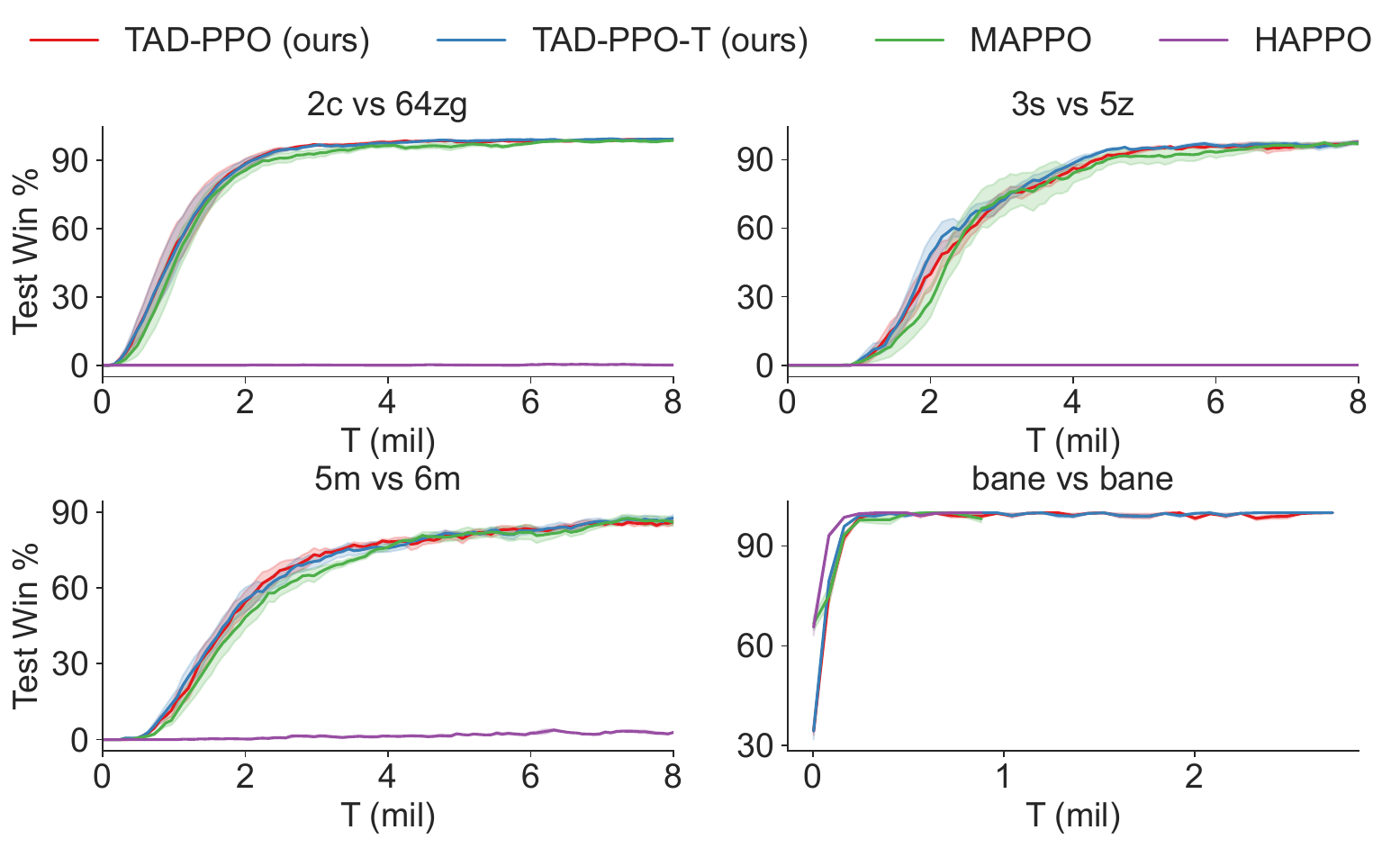}
  \caption{Comparisons between our approach and policy-based state-of-the-art MARL baselines on all \textbf{hard} maps.}
  \label{fig:hard}
\end{figure}

In section \ref{smac}, we have compared and discussed the advantage of our approach against baselines on several representative maps. Here we further compare our reproach against baselines on all maps. The SMAC benchmark contains 14 maps that have been classified as easy, hard, and super hard. In this paper, we design one more map $\mathtt{3h\_vs\_1b1z3h}$, whose difficulty is comparable with official super hard maps.

In Figure. \ref{fig:superhard}, we compare the performance of our approach with baseline algorithms on all super hard maps. We can see that TAD-PPO outperforms all the baselines, especially on $\mathtt{3s5z\_vs\_3s6z}$, $\mathtt{MMM2}$, and $\mathtt{3h\_vs\_1b1z3h}$. These results demonstrate that TAD-PPO can handle challenging tasks more efficiently with theoretical guarantees of global optimality. HAPPO performs poorly on 4 out of 6 super hard maps, demonstrating its limitation on complex tasks.

\begin{figure}[!htbp]
  \centering
    \includegraphics[width=.85\linewidth]{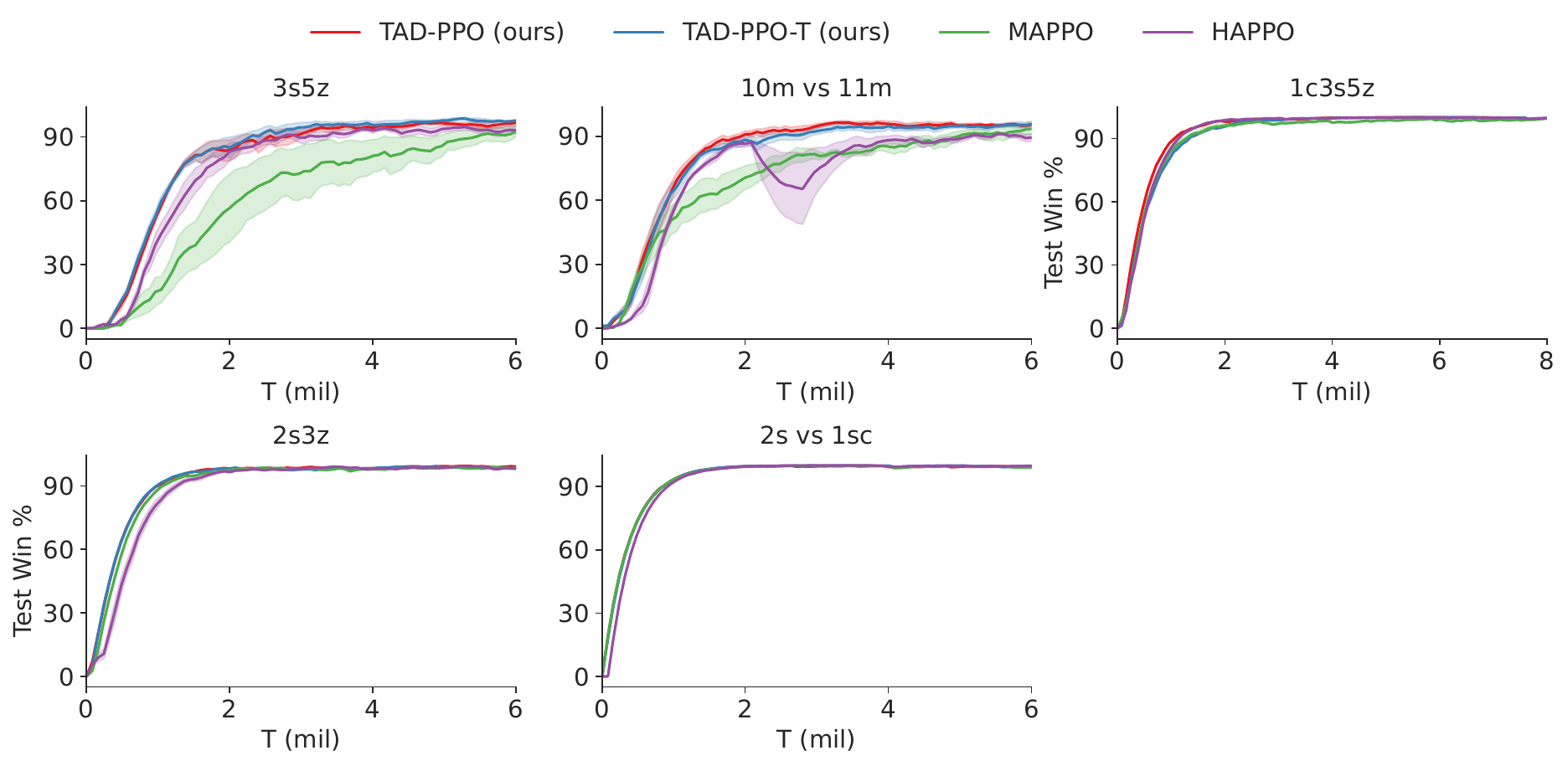}
  \caption{Comparisons between our approach and policy-based state-of-the-art MARL baselines on all \textbf{easy} maps.}
  \label{fig:easy}
\end{figure}


Our approach achieves similar convergence performance with MAPPO in hard maps as shown in Figure.~\ref{fig:hard}, because there is no obvious local optimal point in hard maps. In easy maps, our approach achieves faster convergence rate compared with MAPPO in three out of five maps as shown in Figure.~\ref{fig:easy}.

\subsection{Hyper-parameters}

Our code is implemented based on MAPPO. We share the same structure and hyper-parameters with MAPPO~\citep{yu2021surprising} except improvement we mentioned in Section \ref{sec:TAD-PPO-structure} to instantiate our transformation framework, which contains two extra hyper-parameters: the number of our multi-head attention modules' heads ($N_h$) and the dimension ($D$) of WQs, WKs, and Vs. In all easy and hard maps of the SMAC benchmark, we set $N_h=4$ and $D=4$. In super hard maps of the SMAC benchmark, we adapt two sets of hyper-parameters to deal with various diversity level of teammate strategies. In $\mathtt{MMM2}$, $\mathtt{corridor}$, and $\mathtt{3h\_vs\_1b1z3h}$, we set $N_h=4$ and $D=8$. In other three super hard maps, we set $N_h=3$ and $D=4$. In all GRF academic scenarios used in this paper, we set $N_h=4$ and $D=6$.

\subsection{Running time of TAD-PPO} \label{app:tadppo_running_time}

The sequential update does take longer time than the concurrent update in the training phase, while in the testing phase, our algorithm doesn't take extra time since a decentralized policy is already calculated by distillation.

Specifically, in the training phase, our framework takes $n$ times the time to do action inference, where $n$ is the number of agents. However, it's worth mentioning that, the time of action inference is only part of the time doing a whole training iteration, which also includes environment simulation and policy training. On the whole, the training time cost of our framework is 0.91 times more than MAPPO in SMAC environment $\mathtt{3m}$ (3 agents), and 1.76 times more in SMAC environment $\mathtt{10m\ vs\ 11m}$ (10 agents). This result embodies a trade-off between training time and training performance. Specific time of each part is shown Table \ref{table:time}.

\begin{table}[t]

\caption{Comparison between TAD-PPO and MAPPO on training time}
\centering
\fontsize{7.5pt}{5pt}
\begin{tabular}{|c|c|c|c|c|c|}

\hline
\textbf{Env: } & \textbf{Action}  & \textbf{Env}  & \textbf{Env interaction:}  & \textbf{Policy} & \textbf{The whole training phase:} \\
 \textbf{3m (1M)} & \textbf{inference:} $t_I$ & \textbf{simulation:} $t_S$ & $t_I+t_S$& \textbf{Training:} $t_P$  & $t_I+t_S+t_P$ \\
\hline
 TAD-PPO & 1694.6s  & 915.4s  & 2610s & 233.7s & 2843.7s\\

 MAPPO & 434.9s  & 904.5s & 1339.4s  & 151.9s  & 1491.3s \\

 Ratio & 3.9 & 1.01 & 1.95 & 1.54 & 1.91\\
 \hline
 \hline
\textbf{Env:} & \textbf{Action}  & \textbf{Env}  & \textbf{Env interaction:}  & \textbf{Policy}  & \textbf{The whole training phase:} \\
 \textbf{10m vs 11m (1M)} & \textbf{inference:} $t_I$ & \textbf{simulation:} $t_S$ & $t_I+t_S$& \textbf{Training:} $t_P$  & $t_I+t_S+t_P$  \\
\hline
 TAD-PPO & 5011.7s & 2065.5s & 7077.2s   & 794.1s & 7871.1s  \\

 MAPPO  & 501.1s  & 2063.2s  & 2564.3s & 292.1s & 2856.4s  \\

 Ratio & 10.0 & 1.0  & 2.76  & 2.72 & 2.76 \\ \hline

\end{tabular}
\label{table:time}
\end{table}

\subsection{TAD-DQN results}

\begin{figure}[htb]
  \centering
    \includegraphics[width=13cm]{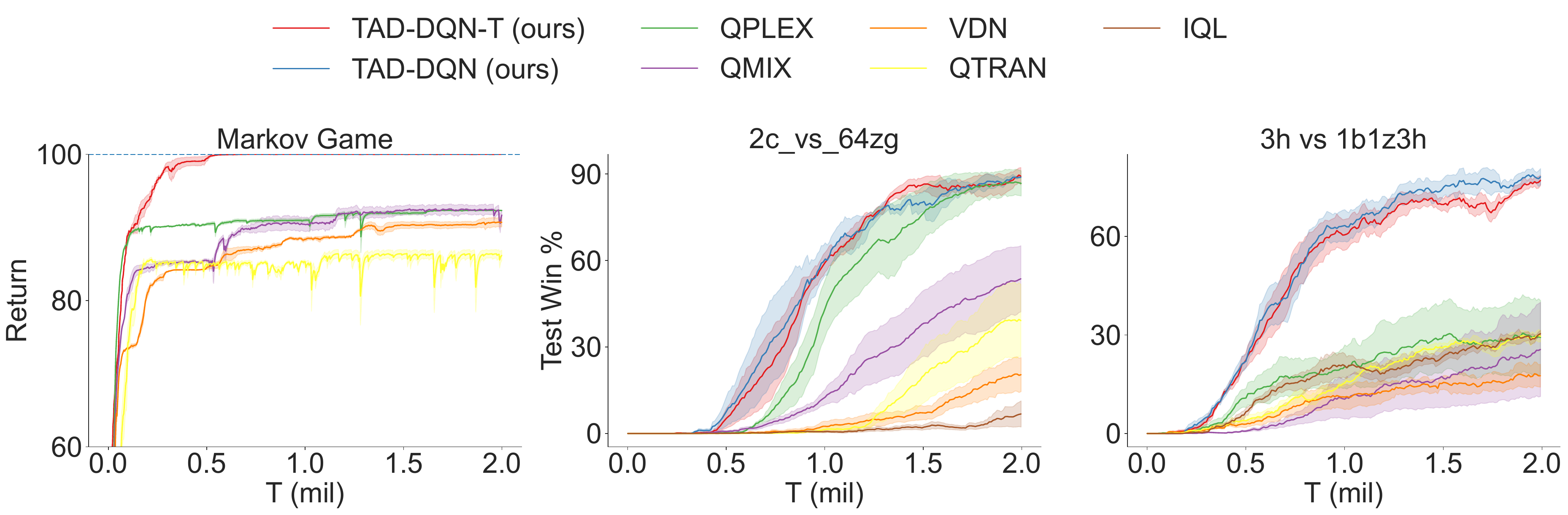}
  \caption{Comparisons between TAD-DQN and off-policy baselines on Markov Game, a \textbf{hard} and a \textbf{superhard} SMAC maps.}
  \label{fig:TDQN-small}
\end{figure}

We also implement the off-policy method TAD-DQN based on the single agent algorithm DQN \citep{dqn} for completeness. In Figure \ref{fig:TDQN-small} We evaluate DQN on a Markov Game (Matrix games with random transition), a \textbf{hard} and a \textbf{superhard} SMAC maps to show that our framework is also compatible with off-policy methods.

\subsection{Learned behaviour of the sequential framework}

We visualize the policy learned by our approach and compare it with MAPPO in MMM2. Based on the comparison, we notice an interesting phenomenon. The joint strategy trained by MAPPO is usually conservative, only moving in a small area, and only two agents are left in the end. On the contrary, the joint strategy trained by our approach is more aggressive. Our agents pull back and forth frequently based on opponents' movement in a large-scale range while ensuring effective fire focus. This phenomenon provides further evidence that our sequential transformation framework enables each agent to fully understand the team strategy for more efficient coordination.

\subsection{Benchmarking with more baselines}
Recently there are some MARL algorithms taking the dependencies among agents into explicit consideration. DFOP \citep{wang2022dfop} factorizes the joint policy to individual policies by learning an independent policy term and a correction term; ACE \citep{li2023ace} considers the bidirectional action dependency. We conduct further experiments with these baselines. Both of the two methods are value-decomposition methods. Therefore it is not fair to compare these methods in terms of either sample efficiency or final performance as the experiments conducted in the main paper, since TAD-PPO is a multi-agent policy gradient algorithm. To enable fair comparison, we compare TAD-PPO with these methods in terms of time efficiency for training. The results are presented in \cref{fig:DFOP-ACE-time-cmp}.

 \begin{figure}[htb]
  \centering
    \includegraphics[width=13cm]{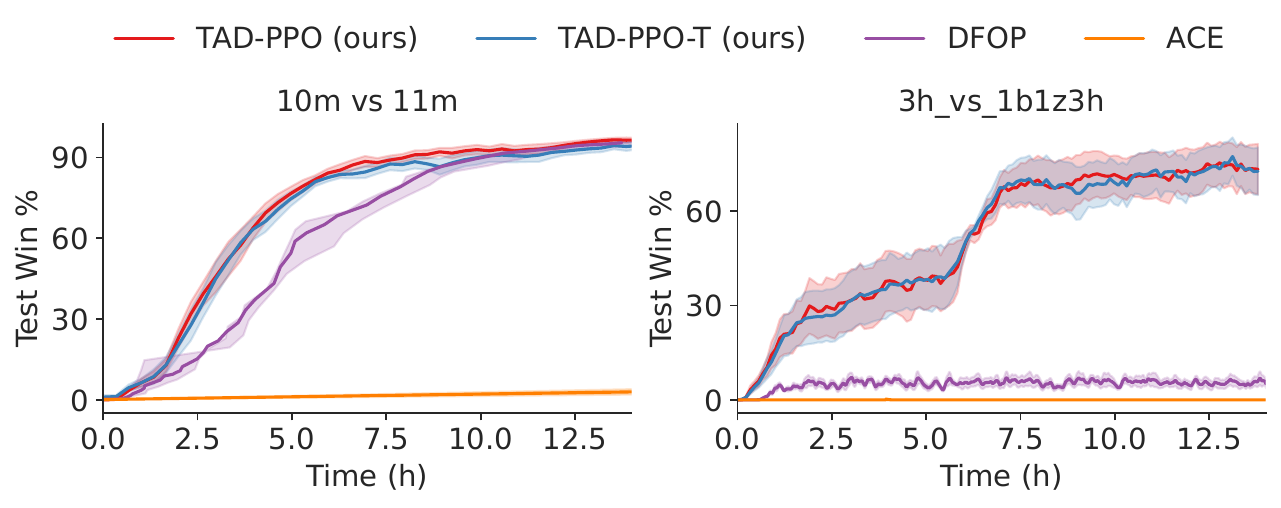}
  \caption{Comparison on time efficiency between TAD-PPO and VD baselines that also takes dependency into consideration on SMAC maps. TAD-PPO converges to a strongly performant solution within 14 hours of training on wall clock, while DFOP and ACE may need a significantly longer time to achieve a well-behaved policy.}
  \label{fig:DFOP-ACE-time-cmp}
\end{figure}

\subsection{Benchmarking on SMAC v2 benchmark} \label{app:exp_smacv2}

SMAC v2 has recently been proposed as an updated version of SMAC benchmark for the evaluation of MARL algorithms. It compensates the lack of stochasticity in its previous version. In this section, we compare the performance of TAD-PPO with one of the state-of-the-art MA-PG representative methods MAPPO, on this new testbed for MARL. TAD-PPO could achieve comparable or even outperformance by a large margin in this stochastic setting compared with MAPPO (\cref{fig:exp_smacv2}).

 \begin{figure}[htb]
  \centering
    \includegraphics[width=13cm]{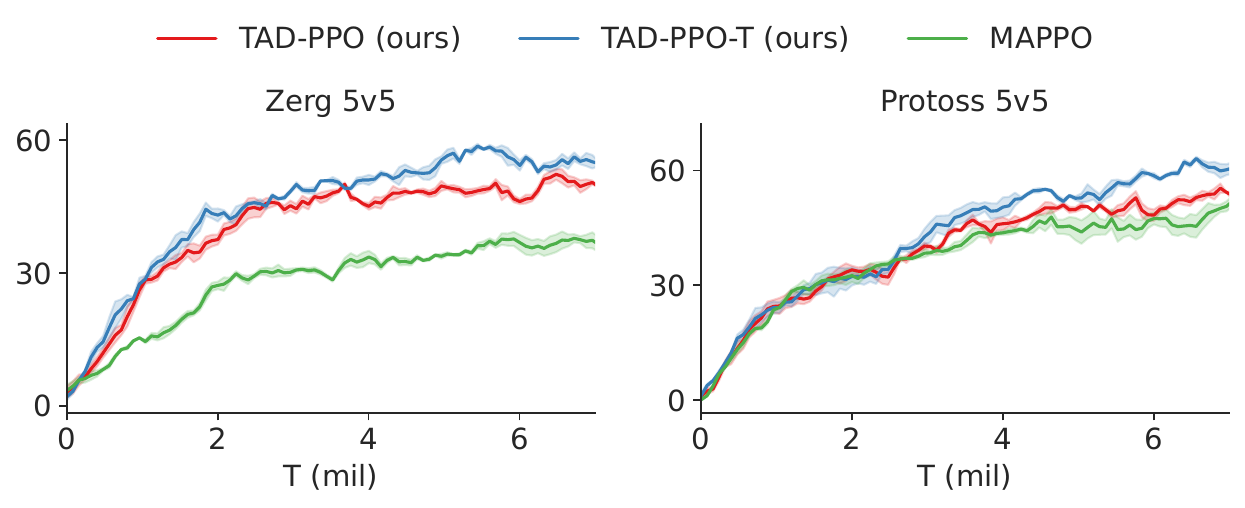}
  \caption{Comparison between TAD-PPO and MAPPO on SMAC v2 maps. SMAC v2 compensates the lack of stochasticity of SMAC benchmark. We present two representative results on SMAC v2. TAD-PPO outperforms MAPPO by a large margin on $\mathtt{Zerg 5v5}$, and it achieves dominant or comparable results on $\mathtt{Protoss 5v5}$ map.  }
  \label{fig:exp_smacv2}
\end{figure}

\section{Extended Didactic Example in Multi-Step Scenario} \label{app:hallway_didactic}
In this section, we present another didactic example of task \texttt{Hallway} (\cref{fig:hallway}). \texttt{Hallway} is a representative of multi-agent coordination tasks in multi-step scenarios compared with the single-step matrix game. This example aims to demonstrate that the state-of-the-art VD and MA-PG methods can be trapped at a suboptimal solution while TAD-PPO with the proposed TAD framework can get out of the suboptimal point and converge to the global optimal point.

\begin{figure}[H]
    \centering
    \begin{minipage}{0.48\textwidth}
      \begin{tikzpicture}
            \draw (0,0) rectangle (2,2);
            \foreach \i in {0,...,2} {
                \draw (\i+2,0) rectangle (\i+1+2,1);
                \draw (\i+2,1) rectangle (\i+1+2,2);
            }
            \draw (5,0) rectangle (7,2);
            
            \node at (4.5,2.5) {$s_0$};
            \node at (1,2.5) {$g_1$};
            \node at (6,2.5) {$g_2$};
            
            \node [fill=blue, opacity=0.5] at (1,1) {+8};
            \node [fill=yellow, opacity=0.5] at (6,1) {+5};
            
            \node  at (4.5,1.3) {\small A1};
            \node  at (4.5,0.3) {\small A2};
            
            \node at (3.5,-0.5) {Actions: Move Left / Move Right};
            \node at (3.5,-1) {Cost: -1 per move};
            \node at (3.5,-1.5) {Initializing at: A1, A2};
        \end{tikzpicture}
        \caption{Illustration of the Hallway Task. The task requires the coordination of two agents to reach the goal positions simultaneously. $g_1$ is a global optimal point with a cumulative reward of +5; $g_2$ is a suboptimal point with a cumulative reward of +4 (considering the movement cost).}
        \label{fig:hallway}
    \end{minipage}
    \hfill
    \begin{minipage}{0.48\textwidth}
        \includegraphics[width=\linewidth]{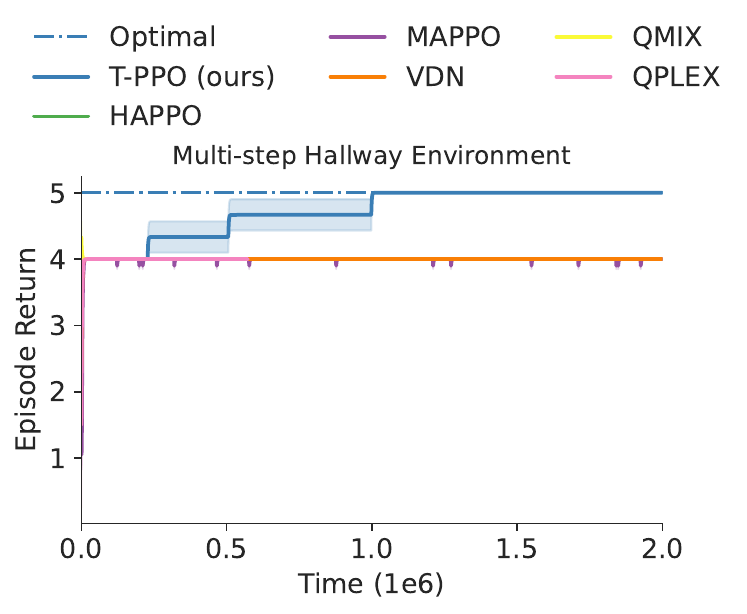}
        \caption{Performance comparisons of TAD-PPO with state-of-the-art VD and MA-PG baselines.}
        \label{fig:exp_hallway}
    \end{minipage}
    \label{fig:two_images}
\end{figure}

In task \texttt{Hallway}, two agents are initialized at $s_0$. At each time step, agents can observe the positions of both agents (full observation) and choose an action to move left or right at a movement cost of -1. There are two terminal states $g_1$ and $g_2$ in the hallway. Agents will win if they simultaneously reach the same terminal state (either $g_1$ or $g_2$). Specifically, agents get a shared reward of +8 if they both reach $g_1$ at the same time and +5 if both reach $g_2$. Otherwise, if one agent arrives at the terminal state earlier than the other, the team will not get the terminal reward as a result of miscoordination. Any of these cases will result in an episode termination. The length of the hallway is set to 5 with an episode limit of 100 to prevent agents from moving back and forth endlessly in one episode. TAD-PPO optimizes to the global optimum with an optimal return of 5 in this task while other VD algorithms and MA-PG algorithms get stuck in the local optimal solution with an inferior return of 4 (\cref{fig:exp_hallway}). Evaluation
results are averaged over three random seeds.

    
    
    
    

\section{Limitations and Discussions}
\label{app:limitations_and_discussions}
As promising avenues for future research, we envision the attention from community could be directed towards tackling the followings.

\paragraph{Increased inference time for optimality.} The transformed single-agent MDP trades a portion of action inference time for optimality, requiring the sequential inference for the actions of agents during training. Detailed training time analysis can be found at \cref{app:tadppo_running_time}. 

\paragraph{Action space-horizon tradeoff in the transformed MDP.} TAD framework shifts the complexity from the action space of multiple agents to time horizon, which theoretically keeps the hardness of problem unchanged. (Details and an analysis on the sample complexity of the transformed MDP can be found at \cref{subsec:complexity}.) In this regard, empirical design efforts in MARL could also shift from handling large action space to handling long time horizon in the transformed MDP.

\paragraph{Taking the task structure of multi-agent into design consideration.} TAD is a general framework that does not exploit the task structure of multiple agents. It would be benefitial for future work to do so. For example, TAD's transformation stage essentially does temporal credit assignment. \citet{wang2021towards} find that the credit assignment is crucial in multi-agent setting.  Future empirical work on TAD could exploit the multi-agent structure for specific design of multi-agent credit assignment. 

\paragraph{Ordering of agents.}  TAD currently requires ordering fixed during the learning process. Although different ordering of agents does not change the optimality guarantee, the design of ordering in empirical algorithms could have an impact on training efficiency.



\end{document}